\documentclass[11pt, letterpaper, twoside]{article}
\usepackage[margin=1in]{geometry}
\usepackage[T1]{fontenc}
\usepackage{ulem}

\usepackage[colorinlistoftodos,textsize=tiny
			]{todonotes}
\usepackage[round]{natbib}
\usepackage{array}
\usepackage{url}
\usepackage{booktabs}

\usepackage{makecell}

\usepackage[font={small,sf}]{caption}

\usepackage{amssymb}

\usepackage{bbm}
\usepackage{soul}
\usepackage{tikz}
\usepackage{amsthm}
\usepackage{amsmath}
\usepackage{float}
\usepackage{physics}
\usepackage{graphicx}
\usepackage{setspace}
\usepackage{enumerate}
\usepackage{mathtools}
\usepackage{caption}
\makeatletter
\def\th@plain{%
  \thm@notefont{}
  \itshape 
}
\def\th@definition{%
  \thm@notefont{}
  \normalfont 
}
\makeatother

\newtheorem{theorem}{Theorem}[section]
\newtheorem{proposition}[theorem]{Proposition}

\newtheorem{assumption}{Assumption}
\newtheorem*{definition*}{Definition}
\newtheorem*{example*}{Example}

\usetikzlibrary{calc}

\newcommand{\N}{\mathbb{N}}

\newcommand{\R}{\mathbb{R}}

\renewcommand{\P}{\mathbb{P}}

\makeatletter
\def\blfootnote{\gdef\@thefnmark{}\@footnotetext}
\makeatother

\allowdisplaybreaks

\title{
    Auctioning Time to Mitigate Latency Races: \\ Theory and Evidence from Blockchains
    \blfootnote{The authors are grateful to Edward Felten and Akaki Mamageishvili for helpful comments and discussion, and acknowledge the support of Ripple's University Blockchain Research Initiative.}
    }
\author{
    Agostino Capponi\footnote{Columbia University, Department of Industrial Engineering and Operations Research and Columbia Business School; ac3827@columbia.edu.} \ 
    and 
    Brian Zhu\footnote{Columbia University, Department of Industrial Engineering and Operations Research; bzz2101@columbia.edu.}
    }
\date{}

\usepackage{url}

\begin{document}

\maketitle
\thispagestyle{empty}

\begin{abstract}
High-frequency trading, in both traditional and decentralized markets, induces latency races and redundant order flow as traders spend resources to win time-sensitive opportunities. 
We show that auctioning artificial time priority can redirect resources away from wasteful speed races toward auction payments. While such waste is difficult to measure in traditional markets, blockchain transactions provide transparent records of these competitive costs through observable duplicate submissions. We study the introduction of Timeboost, a time-priority auction mechanism on Arbitrum, a blockchain that batches transactions before settlement on Ethereum, as a natural experiment. We find that redundant transactions decrease and platform revenue increases relative to comparable networks, consistent with our theoretical predictions.

\end{abstract}

\vfill

\clearpage
\pagenumbering{arabic}

\section{Introduction}

Speed is a central determinant of trading outcomes in high-frequency environments. When access to time-sensitive trading opportunities depends on being first to a platform, traders invest in costly technology, connectivity, and order management systems to reduce latency. These latency races create large expenditures with limited social upside as resources are focused toward speed acquisition and redundant order flow creates congestion externalities. Empirical evidence quantifying these races confirms their significance. \citet{AquilinaBudishONeill2022} find latency arbitrage races, where traders rapidly attempt to ``snipe'' profitable opportunities and cancel failed attempts, occur roughly once per minute per symbol in FTSE 100 stocks. However, the investment decisions underlying these races remain unobservable, because speed technology investment in traditional markets, such as co-location services and proprietary infrastructure, is private information and unable to be observed.

Many studies in the market-design literature have  proposed altering the allocation of time priority, for example through frequent batch auctions (\citet{BudishCramtonShim2015}), speed bumps such as on IEX (\citet{Woodward2018IEX}), or Pigouvian taxes on investment technologies (\citet{BiaisFoucaultMoinas2015}). An alternative  approach to mitigating latency races is to sell time priority explicitly.\footnote{While the IEX introduces a 350 millisecond delay for all non-resting orders, this is not a winner-take-all approach and is mainly designed to protect resting orders from being picked off during periods of high volatility.} Instead of rewarding whoever arrives first, a platform can auction an ordering advantage that determines execution precedence for a specified time interval. The rationale behind time-priority auctions is to convert competition over speed into competition over price: traders bid for priority rather than investing in latency reduction technologies. The idea is closely related to broader proposals in market design that replace technological races with price-based allocation of scarce trading privileges, including batch auctions and congestion pricing mechanisms. Whether this reallocation occurs in practice, however, is difficult to test in traditional markets because the underlying investments in speed technology are not directly observable. 

We study this question in blockchains, where the outcomes of latency races are more directly observable and most notably, \textit{costly}. Blockchain-based trading venues generate time-sensitive opportunities known as {maximal extractable value} (MEV), which include centralized and decentralized exchanges (\citet{CapponiJia2025DEX}), cyclic arbitrage between exchanges (\citet{optimisticMEV}), frontruns (\citet{CapponiJiaWang2025MEV}, Lehar and Parlour (2023)), backruns, and liquidations on decentralized lending protocols. 
Participants compete to be first to exploit profitable opportunities, but the nature of competition differs from traditional markets. Improvements to latency see diminished returns relative to traditional electronic markets due to how blockchain networks process transactions. As a result, traders often submit multiple copies of the same transaction in rapid succession, hoping one executes while the rest fail. This shifts competitive resources from infrastructure investment toward transaction fees and redundant submissions.\footnote{See Section 2 for details on how blockchain transaction routing creates this competitive dynamic.}

Unlike traditional electronic platforms where order submission and cancellation is costless, users incur infrastructural costs for failed attempts on blockchains. Moreover, while message data on traditional exchanges is proprietary and may require invoking financial market regulations to acquire, 
failed attempts are public and visible on blockchains. These duplicate submissions constitute an on-chain trail of evidence and provide a measurable analogue of otherwise unobservable costly efforts to achieve latency. We summarize these differences in Table \ref{tab:hft_grid}.

We leverage the introduction of \textit{Timeboost} on the Arbitrum blockchain, as a market-design intervention in this observable environment.\footnote{Layer-2 blockchains execute transactions off the main Ethereum network and periodically settle batched results on-chain, enabling faster and cheaper execution. Most Layer-2 networks use a centralized sequencer to order transactions, creating a setting where time-priority can be explicitly auctioned.} Timeboost auctions an artificial time priority via a sealed-bid second-price auction. Rather than racing through duplicate submissions, traders can compete over priority directly via auction. We interpret this setting as a natural laboratory that introduces an explicit market for time priority into an environment where latency races are actively occurring. The experiment allows us to study a general question: what happens when a platform replaces speed races with price-based allocation of priority?

\vspace{1em}

\begin{table}[!h]
\centering
\begin{tabular}{p{2.5 cm} p{6.5cm} p{6cm}}
\toprule
 & \textbf{Speed technology} & \textbf{Order submission} \\
\midrule
\textbf{Centralized exchanges} 
& Orders propagated via proprietary networks to matching engine $\to$ latency differences translate directly to priority $\to$ \ul{unobservable but costly competition}
& Traders cancel and resubmit orders via message data at no cost $\to$ \ul{observable but costless competition} \\

\midrule
\textbf{Layer-2 blockchains} 
& Orders propagated via public networks to sequencing node; latency advantages are noisy and bring diminished returns $\to$ \ul{low incentive for competition here}
& Traders submit multiple identical transactions; unsuccessful attempts revert and incur fees $\to$ \ul{observable and costly competition} \\

\bottomrule
\end{tabular}
\caption{Dimensions of Latency Competition Across Trading Environments}
\label{tab:hft_grid}
\end{table}

We develop a model in which traders expend costly effort, interpretable as latency investment in traditional markets or duplicate submissions in blockchain markets, to capture time-sensitive opportunities; these investments have convex costs. We show that introducing a time-priority auction reduces equilibrium effort, 
 reallocating resources from wasteful speed competition into platform revenue. The blockchain environment represents a special case where effort costs are approximately linear, as they are gas fees paid from the reversion of non-winning duplicate transactions. Additionally, the user's revert costs are awarded to the blockchain platform, whereas effort costs in general may go to other service providers.\footnote{Revert costs typically go to the blockchain's sequencer node or associated decentralized autonomous organization (DAO). These entities may then choose to burn all or part of said fees for tokenomics purposes.} Having user costs being accrued by the platform as revenue also holds for traditional HFT settings, as many exchanges either own or tightly control the co-location facilities used by HFT entities. Under this setting, we show that despite accruing fewer fees from fewer transactions submitted, the auction revenue more than makes up for it, increasing overall platform revenue.

How much time priority should be allocated, and is allocating time priority to a single trader the most effective way to mitigate latency races? We extend the model to account for auctions that award multiple time advantage slots of the same amount. Along the intensive margin, i.e.\ the size of the time advantage, we show that broadly, total effort decreases and platform revenue increases in the amount of time priority given to the winner(s). Along the extensive margin i.e.\ how many traders receive a time advantage, we show that when the total number of users in the environment is large, increasing the number of winners may further reduce total effort and improve platform revenue.

We empirically test model predictions using transaction-level data across Layer-2 blockchains and an event-study design around Timeboost’s deployment on Arbitrum, measuring latency effort via bursts of identical transactions and comparing outcomes to control networks. We find significant evidence that following adoption, redundant submissions decline and platform revenue increases relative to benchmark networks. We additionally find that non-priority traders experience higher failure rates in contested opportunities. These findings match the model’s predictions.

Our results have implications for the improved design of traditional markets.  Blockchain markets make priority valuations observable through on-chain fees, with Timeboost introducing an explicit market for time priority. While institutional details differ between blockchain and traditional markets—execution speed, asset characteristics, and participant composition—the core economic question is common to both. Can price-based priority allocation substitute for costly speed competition? Our results suggest that explicit priority pricing can reallocate rents from technology investment to platform revenue without diminishing trading activity, providing  evidence relevant to broader market design debates about latency races in high-frequency trading environments.






Our findings contribute to the broader debate on the market design of exchanges accommodating high-frequency trading. First, our model predicts that auctioning time priority reduces latency investment effort, consistent with theoretical arguments for batch auctions, and providing another parallel mechanism to reduce latency arms races. Second, we show that platform revenue increases under the auction. This speaks to concerns that exchanges resist efficiency-improving design because it takes away co-location and speed-technology revenue (\citet{BudishLeeShim2020}). Our empirical analysis provides causal evidence supporting these predictions in a setting where latency competition is observable, and supports time-priority auctions as a parallel option alongside frequent batch auctions to mitigate latency races, with the added benefit that exchange platforms can capture this auction revenue.

\paragraph{Related Literature}

This paper contributes to three bodies of literature. First, we connect research on latency competition and market design in traditional HFT to decentralized trading environments through a common framework of costly time priority acquisition. A large market microstructure literature studies how HFT and latency competition shape trading strategies and market outcomes. \citet{HasbrouckSaar2013} propose measures of low-latency activity on NASDAQ and document how millisecond-scale response strategies affect spreads, depth, and short-term volatility. \citet{OHara2015HFMM} characterizes how technological change and HFT have transformed the organization of trading, emphasizing that speed and fragmentation alter liquidity provision and price discovery. 

Similar interactions have been documented in decentralized trading environments. \citet{DaianEtAl2019} document priority gas auctions on Ethereum-based decentralized exchanges, in which arbitrage bots bid up gas prices to obtain favorable transaction ordering, interpreting these races as a blockchain analogue of HFT latency competition. Subsequent work measures the scale and composition of MEV in DeFi: \citet{WeintraubEtAl2022} quantify MEV extraction in private relay pools, and \citet{TorresEtAl2024} analyze MEV across multiple networks and show that MEV volume and composition on Layer-2s can rival Ethereum’s. \citet{optimisticMEV} identify optimistic MEV as a major source of persistent spam-like probing on Layer-2s. Our paper connects these literatures by placing both traditional HFT and decentralized MEV extraction in a common model of costly time priority acquisition and by treating Layer-2 sequencers as matching engines whose design choices shape latency races.


Second, we contribute to market design literature in both traditional and blockchain financial markets by analyzing an explicit market for time priority and showing how it reallocates surplus from dissipation to revenue. In traditional settings, \citet{BudishCramtonShim2015} propose batch auctions to eliminate continuous-time races for queue priority, while \citet{BiaisFoucaultMoinas2015} study how slow markets and taxes on speed can implement socially optimal investment in trading technology. On blockchains, Flashbots introduced sealed-bid auctions over transaction bundles as a mechanism to internalize MEV and reduce observable priority gas auctions (\citet{Flashbots2020}). \citet{li2023mevhappy} and \citet{mclaughlin2025clvr} propose sequencing rules to reduce MEV based on greedy and volatility-minimizing heursitics, respectively. \citet{zhu2026revert} study how ``revert protection,'' i.e.\ a feature where sequencers simulate transactions and exclude failed ones from the final block, saving users from paying gas fees on failed transactions, can make priority fee auctions more competitive and increase sequencer revenue. \citet{MamageishviliEtAl2023} analyze a primitive version of the Timeboost mechanism that blends bidding and latency to allocate transaction priority. Our paper contributes to this mechanism-design literature by analyzing single-winner time-priority auctions, as implemented by Timeboost, and extend it to the multi-winner case. We show that under certain regimes, multi-winner time-priority auctions may result in better efficiency outcomes. 

Finally, we contribute to empirical work on decentralized exchange infrastructure by providing novel empirical evidence on how auctioning time priority affects order flow duplication and platform revenue. A growing DeFi microstructure literature studies automated market makers (AMMs), arbitrage, and liquidity provision \citet{Lehar2025, CapponiJia2025DEX}. Most closely related to our work is \citet{messias2025expresslane}, which provides a descriptive empirical analysis of Timeboost, documenting centralization of express-lane usage and high revert rates. By contrast, we use a difference-in-differences design to estimate the causal impact of Timeboost adoption on redundant orderflow and platform revenue.





\section{Institutional Details}

This section describes the institutional features that make blockchain-based trading venues a natural laboratory for studying time-priority auctions. We first contrast the infrastructure underlying traditional electronic markets with blockchain platforms, highlighting how differences in connectivity and network architecture determine the form of latency competition. We then explain the structure of Layer-2 blockchains, the types of profit opportunities that drive competition for transaction ordering, and the mechanics of the Timeboost mechanism that we study empirically.

\paragraph{Platform Infrastructure}

In modern electronic markets, trades occur via a centralized matching engine housed within a data center. Market participants connect directly to this engine using dedicated communication infrastructure, often by co-locating servers inside the exchange facility and routing orders through private cross-connects and fiber or microwave links. The exchange timestamps orders at when they enter its network interface, which determines priority. As connectivity is tightly engineered and shared routing uncertainty is minimized, latency differences translate directly into execution priority. As a result, competition centers on capital-intensive investment in communication technology designed to reduce physical transmission delay.

By contrast, transactions submitted to Layer-2 blockchains traverse heterogeneous paths before reaching the ordering intermediary, often called the \textit{sequencer}. Traders typically transmit through remote procedure call (RPC) gateways, load balancers, and peer-to-peer forwarding layers, all of which introduce random propagation and queueing delays outside the user and sequencer’s control. Although the sequencer ultimately processes transactions by timestamp, the arrival time it observes reflects both physical latency and network randomness. In this setting, marginal improvements in transmission may not translate directly into improved execution priority. Thus, as observed empirically (\cite{optimisticMEV},\cite{gogol2025firstspammed}), users adapt by submitting multiple attempts of the same transaction.

\paragraph{Layer-2 Blockchains}

Ethereum’s consensus and blockspace design prioritize decentralization and security, which constrains throughput and leads to congestion-induced fee spikes during active market periods. Layer-2 (L2) networks address this scalability bottleneck by batching and executing transactions off-chain before posting aggregated transaction data to Ethereum for settlement.\footnote{Despite having a centralized sequencer node, L2 networks achieve consensus by either using a dispute period.} This division of labor allows L2s to offer lower fees, higher transaction rates, and faster execution while preserving the settlement assurances of the base layer. Users submit signed transactions to a sequencer, which aggregates and executes them in batches before committing the resulting state to the underlying ``Layer-1'' blockchain. The sequencing step determines which transactions obtain time priority and therefore which traders capture time-sensitive opportunities—such as arbitrage between automated market makers (AMMs), liquidations in lending protocols, or oracle-based price updates.

Layer-2 blockchains implement their own transaction sequencing rules prior to settlement on Ethereum. Most use a \textit{priority gas auction (PGA)} or a \textit{first-come, first-served (FCFS)} sequencing rule. Under the PGA rule, transactions are ordered by a priority fee that users submit along with their transaction; ties in priority fee are commonly broken by the time at which the sequencer received the transaction. While top-of-block MEV (e.g.\ CEX-DEX arbitrage, liquidations) on PGA-using networks is contested over via first-price auctions, backruns and optimistic MEV does not require competition over priority fees, thus falling into the regime where ties are broken by timestamp. Under the FCFS rule, transactions are ordered directly by the time at which the sequencer received the transaction. 

Sequencing rules on Layer-2 networks determine the relative ordering of transactions and therefore which traders capture time-sensitive opportunities. Most L2s use either priority gas auctions (PGAs) in which users bid higher gas fees to obtain earlier queue positions, creating a fee-based form of time priority or a first-come, first-served (FCFS) rule in which the sequencer orders transactions according to observed arrival time via public or private RPC endpoints. Under both FCFS and PGA sequencing, execution outcomes depend on latency races and/or fee bidding. In PGA settings, when two transactions offer the same priority fee, sequencing falls back to arrival timestamp as a tie-breaking rule.

\paragraph{Maximal Extractable Value}

A significant share of economically meaningful trading activity on Layer-2 networks, also called {\it rollups}  because they 'roll up' many transactions into compressed batches posted to Ethereum, is associated with maximal extractable value (MEV):  rents that arise from the ability to influence transaction ordering relative to liquidity, oracle updates, or protocol state. Many common MEV opportunities consist of CEX-DEX arbitrage and DEX-DEX arbitrage. These opportunities are strongly time-sensitive: once a profitable mispricing is exploited, the payoff disappears for all subsequent transactions.

While MEV in the Ethereum mainnet is typically characterized by CEX-DEX arbitrage and sandwich attacks, the MEV landscape for L2s differs substantially: \cite{gogol2026sandwich} finds that sandwich attacks are essentially non-existent on L2. Instead L2s exhibit a distinctive pattern of \textit{optimistic MEV} and \textit{backruns} (\cite{optimisticMEV}, \cite{gogol2025firstspammed}), in which traders repeatedly submit transactions that would be profitable conditional on favorable sequencing or oracle updates, but that frequently revert when the required preconditions fail to materialize. Optimistic MEV generates large volumes of nearly-identical probing transactions whose result is binary: either capturing the opportunity or reverting at a cost. For the trader, reverts burn gas and sequencing fees, while for the venue they consume blockspace and network capacity, producing congestion externalities for other users. Because MEV payoffs are winner-take-all and highly sensitive to ordering, sequencing rules directly determine who captures MEV and who pays revert costs. Under FCFS or PGA sequencing, traders must race to obtain time priority through arrival latency or fee bidding, and in the absence of fine-grained latency tools, duplicate submissions serve as a costly substitute for speed. 


\paragraph{The Timeboost Mechanism}

Timeboost modifies Arbitrum’s default sequencing rule by introducing a sealed-bid, second-price auction for 200 milliseconds of transaction time priority. The mechanism operates at the beginning of each sequencing interval. Traders submit encrypted bids, denominated in Ether, to an off-chain autonomous auctioneer for access to the Express Lane, a short artificial time advantage that allows the winner’s transactions to be placed ahead of the first-come, first-served (FCFS) queue. The auction has a parameterizable reserve price, set to 0.001 ETH at the time of writing. The auction is resolved, with the winner submitted to Arbitrum's on-chain auction contract, where the highest bidder wins the express lane while paying the second-highest bid as an auction clearing price.\footnote{In the case where there is only one bidder, the price paid is zero.} Auction payments are first transferred from the winner's address to the auction contract, and then to a parameterizable beneficiary address, currently set to the Arbitrum DAO treasury at the time of writing.

Time priority lasts for 60 seconds per auction round. For the auction granting Express Lane access from time $t$ to $t+60$ seconds, bidding starts at $t-60$ seconds and closes at $t-15$ seconds, with the remaining 15 seconds used to propagate auction results on-chain. Any transactions submitted via the Express Lane (which only the winner can submit to) are sequenced according to the timestamp at which the sequencer receives the transaction, which we henceforth call the \textit{arrival timestamp}. All other transactions receive an artificial 200-millisecond delay to their arrival timestamp. Winners can also ``resell'' Express Lane access by allowing third parties to route transactions through their address for side payments on- or off-chain.

\section{Model and Results}

We start with a stylized, yet general model of a contest for a common-value time-sensitive trading opportunity of value $V>0$.  Transactions are ordered, e.g.\ by a matching engine or a sequencer, according to a first-come, first-served (FCFS) rule applied to the arrival timestamps. There are $N\ge 3$ traders indexed by $i\in\{1,\dots,N\}$. Exactly one trader can capture the opportunity, and all other traders fail to capture the opportunity. Trader $i$ can choose an effort level of $k_i\geq0$ into latency at a cost of $C(k_i)$, which determines their probability of successfully winning the opportunity relative to the other traders' investments. Specifically, we assume that latency investments result in a Tullock contest over the opportunity, with trader $i$'s probability of winning as proportionate to their investment relative to the total investment:
\begin{gather*}
    \P(\text{trader $i$ wins}\mid k_i,k_{-i}) = \frac{k_i}{k_i+\sum_{-i}k_{-i}}
\end{gather*}

\begin{assumption}
We additionally assume the following on the cost function $C:\R_+\to\R_+$:
\begin{itemize}
    \item $C$ is twice-differentiable, strictly increasing, and convex;
    \item $C$ has log-convex first derivative;
    \item if $C$ is not linear, then $C'(0)$ is sufficiently small.
\end{itemize}
\end{assumption}
Trader $i$'s expected payoff is then
\begin{align*}
u_i(k_i;k_{-i}) = \frac{k_i}{k_i+\sum_{-i}k_{-i}}\cdot V - C(k_i)
\label{eq:baseline-payoff-general}
\end{align*}
We solve for symmetric pure-strategy Nash equilibria in latency technology investment level, which the following proposition characterizes.

\begin{proposition}\label{prop:eq-baseline}
There exists a unique symmetric pure-strategy Nash equilibrium where each trader chooses investment level $k^*>0$ into latency technology.
\end{proposition}

We now introduce an auction for artificial time priority $T$ in the matching engine or sequencer taking place before the arrival of the trading opportunity. The auction format is a sealed-bid second-price auction where the trader with the highest bid pays the second-highest bid to the auctioneer. Traders now participate in a sequential game consisting of the following stages:
\begin{enumerate}
    \item \textbf{Auction:} Trader $i$ submits a bid $b_i$. The highest bidder
    wins the time priority for this opportunity and pays the second-highest bid.
    \item \textbf{Latency Race:} Trader $i$ chooses the level to invest into latency technology based on the results of the auction: $k_w^{(i)}$ if they win, and $k_l^{(i)}$ if they lose.
\end{enumerate}
Instead of the value of a single profitable opportunity, the value $V$ can now represent the total expected value of profitable opportunities occurring during the time window in which this time advantage is to be allocated. We assume that this advantage exponentially tilts the baseline contest win probabilities for those who did not win the auction by $\exp(-k_w\lambda T)\in[0,1]$, where $k_w$ is the auction winner's investment level and $\lambda$ captures the sensitivity of the latency technology to artificial time priority. Under this setup, a loser $j$'s expected payoff in the continuation (latency race) game:
\begin{gather*}
    u_l^{(j)}(k_l^{(j)};k_w,k_l^{(-j)}) = V\left(\frac{k_l^{(j)}}{k_w+k_l^{(j)}+\sum_{-j}k_l^{(-j)}}\cdot\exp(-k_w\lambda T)\right)-C(k_l^{(j)}).
\end{gather*}
Adding up each auction loser's probability of winning the opportunities yields the auction winner's probability of winning and continuation payoff:
\begin{gather*}
    u_w(k_w;k_l^{(j)}) = V\left(1-\frac{\sum_jk_l^{(j)}}{k_w+\sum_jk_l^{(j)}}\cdot\exp(-k_w\lambda T)\right)-C(k_w).
\end{gather*}
The following result characterizes the latency race subgame. 

\begin{proposition}\label{prop:eq-submission}
There exists a unique symmetric pure-strategy Nash equilibrium $(k_l^*,k_w^*)$  in the latency race subgame where $k_l^*$ and $k_w^*$ are the levels invested into latency technology by each trader conditional on losing and winning the auction, respectively, and satisfy $0 < k_l^* < k_w^*$. Moreover, letting $u_l^*$ and $u_w^*$ denote loser and winner equilibrium continuation payoffs, we have $0 < u_l^* < u_w^*$.
\end{proposition}

Introducing an auctioned time advantage alters the equilibrium of the baseline latency race. The ordering $0<k_l^*<k_w^*$ established by Proposition \ref{prop:eq-submission} captures the key behavioral response. Intuitively, conditional on winning the auction, the marginal return to effort is amplified by the time advantage, so the winner invests more in latency. Conditional on losing, the time advantage reverses the slope of marginal returns, so the loser scales back effort.

Both types retain positive continuation value in equilibrium. Even the loser obtains a strictly positive expected payoff, since its probability of winning the contest remains positive absent complete dominance by the advantaged trader. The winner, however, enjoys a strictly larger continuation surplus: at any fixed level of effort, the time advantage increases its contest success probability, and the resulting gain is not fully dissipated by equilibrium effort adjustments. This surplus gap, $u_w^*-u_l^*$, determines willingness to pay for the time advantage in the auction stage. Hence, the mechanism can extract rents from the latency race while preserving participation incentives for individual traders.

\begin{proposition}\label{prop:eq-auction}
Let $b^* = u_w^* - u_l^*$. Then $(b^*,k_l^*,k_w^*)$ comprise the unique symmetric subgame-perfect Nash equilibrium in the full auction-submission sequential game.
\end{proposition}

Proposition~\ref{prop:eq-auction} ties everything together by showing that the unique equilibrium of the latency race subgame endogenously pins down a unique equilibrium of the preceding auction. In the sealed-bid second-price auction, bidding this value is the
(unique symmetric) best response: any lower common bid can be profitably outbid to secure strictly positive surplus, while any higher common bid yields negative expected surplus relative to losing. As a result, $b^*=u_w^*-u_l^*$ is the unique symmetric equilibrium bid, and combining it with the submission strategies $(k_l^*,k_w^*)$ yields a unique symmetric subgame-perfect Nash equilibrium of the full sequential game.

Our key result compares the total amount of investment level expended by traders in equilibrium without and with the time-priority auction, denoted as $K^*$ and $K^*_A$, respectively. In the language of this section, we can express these as
\begin{gather*}
    K^* = Nk^*, \\ K^*_A = k_w^*+(N-1)k_l^*.
\end{gather*}
The following result shows that less total effort is spent to achieve latency under a time-priority auction.

\begin{theorem}\label{thm:comparison-general}
Traders invest less into latency technology under a time-priority auction than without such an auction, i.e.\ $K_A^* < K^*$. 
\end{theorem}

Without an auction, priority is determined entirely by latency, so all traders must invest in speed to remain competitive. This produces a symmetric arms race: each trader’s marginal gain from being slightly faster is offset by others doing the same, and in equilibrium the surplus associated with time priority are largely dissipated through latency investment. By contrast, when time priority is allocated via a single-winner auction, only the winner has strong incentives to invest further in latency, while the remaining traders face a systematic disadvantage and therefore scale back. The reduction in additional latency investment by non-winners more than offsets the additional investment by the winner, lowering aggregate speed expenditure. In effect, the auction compresses a continuous latency race into a one-shot competition for priority, eliminating secondary races among disadvantaged traders and converting dissipative investment into auction payments.

\section{The Blockchain Setting and Platform Revenue}

By modeling transaction submissions and random latency delays, we microfound the blockchain setting as a reduced-form instance of the general model with linear effort costs, linking observable on-chain behavior to latency investment. Here, the competition is over an MEV opportunity of common value $v>0$. A sequencer orders transactions according to a first-come, first-served (FCFS) rule applied to their arrival timestamps. Exactly one trader can exploit the MEV opportunity by being sequenced first with the proper transaction; all other copies revert, failing to capture the opportunity. Traders can submit multiple copies of the same MEV-exploiting transaction, incurring a linear per-copy cost that represents the gas cost of failed attempts. \footnote{We assume that traders submit transactions that are identical in recipient address and calldata, but have increasing nonces, meaning that failed attempts appear on-chain, such as the case for backruns and optimistic MEV.}

Each submitted transaction experiences a random delay before it is observed by the sequencer for inclusion. Formally, the timestamp is shifted by an i.i.d.\ delay $\varepsilon\sim F$ (with density $f$), capturing network latency. Suppose that trader $i$ submits $k_i$ copies of the transaction. Then the effective arrival timestamp relevant under FCFS is the minimum of $k_i$ i.i.d.\ draws from $F$:
\[
X_i(k_i)\;:=\;\min\{\varepsilon_{i1},\ldots,\varepsilon_{ik_i}\}.
\]
Under FCFS, trader $i$ captures the opportunity iff $X_i(k_i)\le \min_{j\neq i}X_j(k_j)$. Intuitively, submitting more copies reduces the minimum delay and increases the chance of being first. We assume that a successful transaction costs $g\in(0,v)$ in gas fees, and reverted transactions must pay a fraction $r\in(0,1)$ of the gas fees for a successful transactions.

For tractability, we assume that latency is exponentially distributed: a single copy has a delay of $\varepsilon\sim \mathrm{Exp}(\lambda)$ with rate $\lambda>0$. Then we have that
\[
X_i(k_i)\sim \mathrm{Exp}(k_i\lambda).
\]
In practice, $k_i$ should be a nonnegative integer, but for tractability, we relax this and allow traders to choose fractional amounts for the number of copies to send, expanding their action space to $k_i\in\R_+$. While this breaks the interpretation that $X_i(k_i)$ is the minimum of several i.i.d.\ draws from $F$, we still have that $X_i(k_i)\sim \text{Exp}(k_i\lambda)$.
Trader $i$'s expected payoff is then
\begin{align*}
u_i(k_i;k_{-i}) &= (v-(1-r)g)\cdot \Pr\!\left(X_i(k_i)\le \min_{j\neq i}X_j(k_j)\right) -rgk_i \\
&\coloneqq V\cdot \Pr\!\left(X_i(k_i)\le \min_{j\neq i}X_j(k_j)\right) -Ck_i
\label{eq:baseline-payoff-general}
\end{align*}
where we reparameterize with $V=v-(1-r)g$ and $C=rg$.\footnote{Note that the cost accounting is done by adding the revert costs for all transactions plus an additional gas fee if the transaction is successful.} Under independent exponentially-distributed latency delays, the win probability simplifies to
\begin{gather*}
    \Pr\!\left(X_i(k_i)\le \min_{j\neq i}X_j(k_j)\right) = \frac{k_i}{k_i+\sum_{j\neq i}k_j},
\end{gather*}
resembling the Tullock contest-like win probability in the general model. Under linear costs, the equilibrium number of copies has a simple closed form:
\begin{gather*}
        k^* = \frac{N-1}{N^2}\cdot\frac{V}{C}.
\end{gather*}

Using this expression, we can compute important equilibrium quantities. Let $K^*$ be the total number of transactions submitted in equilibrium. This represents the level of on-chain spam, since $K^*-1$ of these transactions will revert on-chain, and has formula
\begin{gather*}
    K^* = Nk^* = \frac{N-1}{N}\cdot\frac{V}{C}.
\end{gather*}
We can also compute an individual trader's equilibrium utility $u_i^*$, and how the total value $V$ is split between users, captured by the term $W^*$, and revenue $R^*$ accrued by the sequencer/DAO from reverted gas fees:\footnote{To be precise in the case of Arbitrum, gas fees not spend on posting block data to the Ethereum mainnet go to the Arbitrum DAO as well.}
\begin{gather*}
    u_i^* = \frac{V}{N}-Ck^* = \frac{V}{N^2}, \\ 
    W^* = Nu_i^* = \frac{V}{N}, \\
    R^* = V-W^* = \frac{N-1}{N}\cdot V.
\end{gather*}
These corollaries highlight the core inefficiency motivating Timeboost-style mechanisms: competition to be first under FCFS induces redundant transaction copies. In later sections we show how auctioning a time advantage changes the competitive margin, reducing $K^*$ and reallocating surplus away from revert-driven waste and toward auction payments.

We now consider the time-priority auction for a time advantage of $T$ in the blockchain setting. As in the baseline model, each submitted copy experiences an exponential idiosyncratic latency delay before it is observed by the sequencer. Let $X_w(k_w)\sim \mathrm{Exp}(k_w\lambda)$ denote the winner’s earliest timestamp and $X_\ell^{(j)} \sim \mathrm{Exp}(k_l\lambda)$ denote loser $j$’s earliest timestamp, with being independent across arbitrageurs. Now, the winner’s effective arrival time is $X_w - T$, while each loser’s effective arrival time is still $X_\ell^{(j)}$. Therefore the winner succeeds if and only if
\[
X_w - T \le \min_j X_\ell^{(j)}.
\]
Similarly, a particular loser $j$ succeeds if and only if
\[
X_\ell^{(j)} < \min_{-j} X_\ell^{(-j)} \ \text{and}\ \ X_\ell^{(j)} \le X_w - T.
\]
Using standard probability calculations and exponential distribution properties, the continuation payoffs for the winner and loser are, respectively:
\begin{gather*}
    u_w(k_w;k_l^{(j)}) = V\left(1-\frac{\sum_jk_l^{(j)}}{k_w+\sum_jk_l^{(j)}}\cdot\exp(-k_w\lambda T)\right)-Ck_w \\
    u_l^{(j)}(k_l^{(j)};k_w,k_l^{(-j)}) = V\left(\frac{k_l^{(j)}}{k_w+k_l^{(j)}+\sum_{-j}k_l^{(-j)}}\cdot\exp(-k_w\lambda T)\right)-Ck_l^{(j)},
\end{gather*}
coinciding with those of the general model.

Evaluating the same important quantities under Timeboost, we first have that the total number of transactions submitted in Timeboost equilibrium, denoted $K^*_A$, is
\begin{gather*}
    K^*_A = k_w^*+(N-1)k_l^*.
\end{gather*}
Since the winner pays $u_w^*-u_l^*$ in Timeboost equilibrium, all arbitrageurs receive an equilibrium payoff of $u_l^*$. The division of the surplus value $V$ between arbitrageurs and sequencer, denoted by components $W^*_A$ and $R^*_A$, respectively, is
\begin{gather*}
    W^*_{K^*} = Nu_l^*, \\ 
    R^*_{K^*} = V-W^* =  V - Nu^*_l.
\end{gather*}

\begin{proposition}\label{prop:comparison}
Arbitrageurs send fewer total transactions and the sequencer earns more revenue in equilibrium under Timeboost than in the baseline setting, i.e.\ $K^*_{tb} < K^*$ and $R^*_{tb} > R^*$.
\end{proposition}

In the blockchain setting with linear costs, latency investment manifests as duplicate transaction submissions. The efficiency gain from the time-priority auction corresponds to a reduction in redundant submissions, freeing blockspace that would otherwise be consumed by failed race attempts. Interestingly, when considering sequencer/DAO revenue, one would expect there to be a tradeoff between auction revenue and less gas fees from reverted duplicate transactions, but Proposition \ref{prop:comparison} shows that equilibrium revenue unambiguously increases under a time-priority auction. Hence, this result indicates that time-priority auctions in blockchains, such as Timeboost, convert what would have been dissipated in revert-related costs for incremental latency improvement into direct payments for a substantial ``latency improvement'' via artificial time advantages.

\section{Time-Priority Auctions with Multiple Winners}

The analysis thus far has focused on a single-winner auction that allocates a single time advantage. We now generalize the mechanism to allow multiple winners, each receiving the same time advantage. This extension allows us to study how the allocation of priority along the extensive margin---the number of traders granted priority---interacts with the intensive margin---the size of the time advantage, in shaping equilibrium user effort and platform revenue. Additionally, holding the number of winners fixed, we examine how key equilibrium outcomes vary with the size of the time advantage. By varying the number of priority slots and the magnitude of each priority slot, we can characterize how competition, effort, and revenue respond to broader access to priority and more concentrated priority.

Let $n_w,n_l\in\N$ denote the number of winners and losers, respectively, with $n_w+n_l=N$. We now assume that traders participate in a (sealed-bid) \textit{generalized second-price auction} when there are multiple winners. In such an auction, the $n_w$ highest bids receive the time advantage, and the $i$-th highest bidder pays the value of the $(i+1)$-th highest bid. While this auction format is not dominant-strategy incentive compatible, it is the natural extension of the sealed-bid second-price auction used in the single-winner model and admits weak Nash equilibria where each winner pays the continuation value of having priority, similarly to the single-winner case.

Letting $K_w=k_w^{(i)}+\sum_{-i}k_w^{(-i)}$, $K_l=k_l^{(j)}+\sum_{-j}k_l^{(-j)}$, and $K=K_w+K_l$. Under exponential random delays, continuation payoffs are given by
\begin{gather*}
    u_w^{(i)}(k_w^{(i)};k_w^{(-i)},k_l^{(j)}) = V\cdot\frac{k_w^{(i)}}{K_w}\left(1-\frac{K_l}{K}\cdot\exp(-\lambda T K_w )\right)-Ck_w^{(i)}, \\
    u_l^{(j)}(k_l^{(j)};k_w^{(i)},k_l^{(-j)}) = V\left(\frac{k_l^{(j)}}{K}\cdot\exp(-\lambda TK_w)\right)-Ck_l^{(j)}.
\end{gather*}
Similarly to the previous section, we search for interior type-symmetric equilibria $k_l^*,k_w^*>0$ where all auction losers choose $k_l^*$ and all auction winners choose $k_w^*$ in the resulting continuation game. We show that if such an interior type-symmetric equilibrium exists, then $0<k_l^*<k_w^*$ and $0<u_l^*<u_w^*$, extending Proposition \ref{prop:eq-submission} to the multi-winner case. This places the ``value'' of the time-priority slot at $u_w^*-u_l^*$, and a standard auction theory argument shows that the symmetric bid profile $b^*=u_w^*-u_l^*$ constitutes a Nash equilibrium in the generalized second-price auction. The combined profile $(b^*,k_l^*,k_w^*)$ is thus \textit{type-symmetric} a subgame-perfect Nash equilibrium.

Our key results in the section concern the analysis of the multi-winner setting along the intensive and extensive margins. Proposition \ref{thm:fix-n-opt-t} extends the effort and revenue results from the single-winner environment to the multi-winner setting. The same economic forces apply here: time-priority becomes a more reliable determinant of winning the trading opportunity, lowering the marginal value of additional submission intensity. We expand on this result by examining total effort and platform revenue as the size of the time advantage itself varies. 

The comparative statics on effort and revenue are globally monotonic in advantage size in the single-winner case, as more time priority further increases the value of the priority slot, redirecting more financial resources into the auction rather than effort. With multiple winners, we show that there exists an interval $[0,\bar T)$ such that if $T\in[0,\bar T)$, a unique type-symmetric equilibrium exists. Along this interval, equilibrium total effort decreases and platform revenue increases in $T$, indicating that some advantage is always better than none in terms of mitigating latency races and capturing revenue.

\begin{proposition}\label{thm:fix-n-opt-t}
Suppose that user costs are linear and contribute to platform revenue. For any number of winners, there exists a range of time advantage sizes for which a unique interior type-symmetric equilibrium $(k_l^*,k_w^*)$ exists. In this range, increasing the priority advantage reduces total effort and raises platform revenue.\footnote{In the single-winner case, a unique type-symmetric equilibrium exists for any $T\geq0$. Equilibrium total effort decreases and platform revenue increases monotonically in $T$.}
\end{proposition}

Proposition \ref{thm:fix-t-opt-n} characterizes the extensive margin of time-priority allocation. Increasing the number of winners spreads priority across more traders, presenting an interesting tradeoff of incentives. On one hand, competition between the loser and winner types is further dampened as auction losers are facing multiple users with time advantages, rather than just one. Yet, on the other hand, allowing for multiple winners now introduces competition between winners themselves. Although priority becomes less exclusive, the aggregate demand for priority remains high because more traders can benefit from obtaining it. In equilibrium, over a range $n_w\in[1,\bar n_w)$ where the symmetric equilibrium is well behaved, the former force dominates the latter, resulting in a decline in total effort and increase platform revenue. 

\begin{proposition}\label{thm:fix-t-opt-n}
Suppose that user costs are linear and contribute to platform revenue. When the total number of users in the environment is large, for any time advantage size, there exists a range of winner counts for which a unique type-symmetric equilibrium exists. In this range, increasing the number of winners reduces total effort and raises platform revenue.
\end{proposition}

Together, the propositions highlight distinct roles for the intensive and extensive margins. Increasing the size of the time advantage primarily weakens incentives to race conditional on priority allocation, while increasing the number of winners redistributes competitive pressure across traders. Both channels reduce wasteful competition and convert rent dissipation into platform revenue, but they operate through different economic mechanisms: the intensive margin changes the \textit{saliency} of priority, whereas the extensive margin changes its \textit{scope}.

\section{Empirical Analysis}

This section evaluates the model’s predictions using the introduction of Timeboost as an empirical implementation of a time-priority auction. As a market-design intervention that converts latency competition into explicit price competition for execution priority, the mechanism therefore provides a setting in which the allocation of time priority changes while trading opportunities and participant incentives remain otherwise comparable, allowing us to study how equilibrium behavior responds when priority is priced rather than obtained through racing.

Blockchain markets make this exercise feasible because latency competition is observable. In traditional HFT environments, investments are privately-held information and unobservable to researchers. By contrast, on blockchain platforms, attempts to win latency races manifest as bursts of identical transactions submitted within short intervals. We interpret these duplicate submissions as the on-chain analogue of speed investment and use them as a proxy for effort. Our model predicts that a time-priority auction should therefore reduce redundant submissions and redirect competition toward payments for priority.

Timeboost implements precisely such a change. Prior to adoption, transaction ordering on Arbitrum followed a first-come, first-served rule, under which traders competed by submitting multiple copies of the same transaction to increase the probability of early inclusion. Timeboost introduces a sealed-bid auction that grants a short execution advantage to winning bidders, replacing pure arrival-time competition with an explicit price-based allocation of priority. In the context of the model, this corresponds to shifting from a race for execution probability to an auction for execution rights.

We study this transition using transaction-level data across multiple Layer-2 blockchains as control groups and an event-study design around Timeboost’s deployment. Despite other Layer-2 blockchains ordering transactions via priority gas fee, ties among priority fees are still broken by timestamp, and many MEV opportunities on Layer-2 networks rely on being the first to exploit it in the zero priority fee regime. This results in a latency race in those settings as well for these MEV opportunities. By comparing outcomes on Arbitrum to other Layer-2 networks that did not introduce time-priority auctions, we isolate how explicitly pricing execution priority affects equilibrium behavior. There are three testable implications of our model:
\begin{enumerate}
    \item Timeboost decreases transaction spamming. (Theorem \ref{thm:comparison-general})
    \item Timeboost increases sequencer/DAO revenue. (Theorem \ref{prop:comparison})
\end{enumerate}
To test these implications, we conduct a panel data analysis using data across Layer 2 blockchains and a within-Arbitrum analysis.

\subsection{Data Collection}

We assemble a panel of on-chain activity for Arbitrum and a set of major EVM L2 comparison networks using Dune Analytics. We collect data from the Arbitrum, Optimism, Polygon, Base, and Avalanche (C-Chain) blockchains between February 17 and June 17 of 2025, a two-month window before and after the deployment of Timeboost, which happened on April 17, 2025.
Our objective is to measure spamming and wasteful competition among bots and searchers. This phenomenon is relevant both in FCFS sequencer settings, where arrival timing drives inclusion, and PGA environments, where actors repeatedly submit or adjust transactions to win priority. 

To measure spamming in a comparable way across chains, we look for bursts of near-identical transactions within short second-level windows. Setting a window length of 2 seconds, we track the transaction's sender, recipient, value, function selector, and calldata. If transactions with these same characteristics appear multiple times in the same burst window, we label all transactions beyond the first as repeated (redundant) attempts, interpreting them as a conservative proxy for automated resubmission and competitive duplication. This is motivated by common bot execution logic: when inclusion and ordering are uncertain, whether due to FCFS latency dispersion, PGA-style priority competition, or optimistic MEV transactions \citet{optimisticMEV}.


A key challenge is that the objects of interest, i.e.\ attempts to win ordering priority, are not easily identifiable in blockchain data. Unlike traditional exchanges, on-chain transactions do not have a well-labeled field identifying whether an action corresponds to arbitrage, liquidation, and so on. In particular, economically identical actions can be implemented through arbitrary contracts. Users frequently deploy custom contracts, proxy contracts, and externally-owned accounts (EOAs) that encode trading logic internally. Thus, many MEV-related transactions can be executed with minimal or even empty calldata (i.e.\ a call that triggers internal state-dependent execution). Since attempting to manually classify MEV-related transactions using contract allowlists or application binary interface (ABI) decoding is difficult and could exclude economically relevant activity, biasing measurement toward more well-established protocols, we adopt a protocol-agnostic approach: rather than focusing on a narrow set of known MEV contracts or function calls, we identify behavior consistent with a latency race, i.e.\ multiple submissions of the same economic action within a short interval.

This broader classification introduces potential false positives, such as automated resubmissions unrelated to strategic activity. Importantly, however, this noise is not specific to the treated chain: since the same data collection procedure is applied uniformly across all Layer-2 blockchains in the panel, any automation or wallet retry logic enters both treatment and control groups. In the difference-in-differences design, this noise will be absorbed by chain and time fixed effects. The key empirical identification strategy relies on differential changes in outcome variables of interest \textit{after} the introduction of Timeboost, rather than on the level of repeated activity. Remaining misclassification would mainly add noise, in fact attenuating estimates, rather than creating a false treatment effect.

Recent evidence supports interpreting repeated failed transactions on rollups as strategic competition. \cite{optimisticMEV} find that in L2s, searchers submit transactions that repeatedly probe liquidity pools for cyclic arbitrage opportunities, executing the trade if profitable and reverting the transaction otherwise; these types of transactions are responsible for up to 50\% of all gas fees on some networks (e.g.\ Base and Optimism). \citet{gogol2025firstspammed} find similar empirical evidence of spamming in L2s as an arbitrage strategy, and further document that searchers often prefer submitting duplicate transactions over bidding priority fees, even on networks using PGAs, indicating competition over arrival time rather than price. Their findings support that bursts of near-identical transactions within short time intervals---the behavioral signature that our data collection process leverages---is indicative of MEV-related competition on L2s. 

Another important nuance in the MEV landscape for L2s is how certain types of MEV differ depending on whether a trader can replace a pending transaction or must submit independent attempts. In ``top-of-block'' MEV (e.g.\ arbitrage or frontruns), traders repeatedly update a single transaction using the same nonce while adjusting timing or fees; each submission replaces the previous one, so only one execution can occur and losing attempts leave little on-chain footprint. In ``rest-of-block'' MEV (e.g.\ backruns, optimistic arbitrage), outcomes depend on realized block state, so replacement is infeasible and traders submit multiple transactions with increasing nonces. Thus, unsuccessful attempts revert and remain recorded, resulting in costly duplication. Because this distinction holds across both PGA and FCFS sequencing rules, our data measure a consistent object across networks: observable contest effort rather than protocol-specific behavior. These differences are summarized in Table \ref{tab:mev_grid}.

\vspace{1em}

\begin{table}[!h]
\centering
\begin{tabular}{l p{5.5cm} p{5.5cm}}
\toprule
 & \textbf{Top-of-block MEV} (e.g.\ arbitrage, frontruns)
& \textbf{Rest-of-block MEV} (e.g.\ backruns, optimistic MEV) \\
\midrule
\textbf{PGA L2s} & Contested via fee bidding wars; failed attempts are costless and leave no on-chain record & Contested via multi-nonce spamming with zero priority fee; failed attempts are costly and leave an on-chain record \\
\midrule
\textbf{FCFS L2s} & Contested via same-nonce spamming; failed attempts are costless and leave no on-chain record & Contested via multi-nonce spamming; failed attempts are costly and leave an on-chain record \\
\bottomrule
\end{tabular}
\caption{MEV Opportunity Types Across Execution Environments}
\label{tab:mev_grid}
\end{table}

Our empirical strategy focuses on this observable component of latency competition. While this does not account for all MEV activity, it measures the subset for which competition requires multiple independent submissions and therefore appears on chain. Crucially, this subset exists both on Arbitrum and on comparison Layer-2 networks, where ordering among these transactions is likewise determined by arrival time. The identification therefore compares how traders adjust observable effort when priority is priced instead of obtained through racing.



We first record the number of repeated transactions and the gas consumed by failed repeated transactions, denoted by $\textsf{RepTXs}_{c,t}$ and $\textsf{RepGas}_{c,t}$, where $c$ indexes a blockchain and $t$ indexes the time period. On Arbitrum after the implementation of Timeboost, the auction proceeds also contribute to the revenue. We thus collect data for each Timeboost auction round within our sample window, focusing on the payment, denoted by $\textsf{AuctionPayment}_t$, of the winner (the second-highest bid). The $\textsf{Revenue}_{c,t}$ variable then integrates revenue from gas fees and auctions across Layer 2s:
\begin{gather*}
    \textsf{Revenue}_{c,t} = \textsf{RepGas}_{c,t}+\textsf{AuctionPayment}_t\cdot1_{\{c=\textsf{Arbitrum}\}}.
\end{gather*}
We collect a set of controls, including base gas fees and DEX volume on each blockchain, as well as returns and volatility of Ether (relative to US dollars). While Arbitrum, Optimism, and Base use ETH for gas payments, Polygon and Avalanche use native tokens; to account for this, we convert values into ETH. All variables are aggregated at daily frequencies.


\vspace{1em}

\begin{table}[h]
\centering
\renewcommand{\arraystretch}{1.25} 
\begin{tabular}{l r r r r r r r r}
\hline
 & \multicolumn{4}{c}{\textsf{RepTXs}} & \multicolumn{4}{c}{\textsf{Revenue}} \\
\cline{2-5}\cline{6-9}
Chain & \textbf{Mean} & SD & Min & Max & \textbf{Mean} & SD & Min & Max \\
\hline
Arbitrum     & 68{,}416  & 92{,}305  & 2{,}850  & 643{,}983 & 57.14 & 78.64 & 0.322 & 592.3 \\
Avalanche    & 2{,}774   & 4{,}625   & 304      & 39{,}127  & 0.026  & 0.046  & 0.001 & 0.262 \\
Base         & 216{,}987 & 142{,}519 & 62{,}578 & 857{,}769 & 0.503  & 0.879  & 0.108 & 8.215 \\
Optimism     & 43{,}092  & 56{,}968  & 2{,}701  & 516{,}531 & 0.012  & 0.017  & 0.002 & 0.156 \\
Polygon      & 11{,}290  & 7{,}686   & 2{,}577  & 36{,}718  & 0.001  & 0.002  & 0.000 & 0.009 \\
\hline
\end{tabular}
\caption{\centering Summary statistics by chain for total repeated transactions and sequencer revenue.}
\label{tab:summ-stats}
\end{table}

Summary statistics in Table \ref{tab:summ-stats} show substantial heterogeneity across chains in all three metrics of interest. Base and Arbitrum exhibit the highest average repeated transaction counts, indicating more activity and frequent competition over MEV, while Polygon and Avalanche are far lower in levels. Failure rates also differ markedly: Arbitrum and Polygon have relatively high mean failure rates while Optimism stands out with a much lower mean failure share. Revenue varies by orders of magnitude as well, as Arbitrum’s mean is far larger than the other chains, likely due to higher base gas fees and the magnitude of transactions processed. 




\subsection{Cross-Network Analysis}

We estimate the causal effect of Timeboost on transaction spamming and sequencer revenue using a two-way fixed-effects (TWFE) design that compares changes on Arbitrum (the treated chain) before and after the Timeboost launch to contemporaneous changes on a set of comparison L2 networks with no major changes to their sequencer ordering rules. The identifying variation is thus cross-sectional (treated vs. control chains) interacted with time-series variation (post- vs. pre-implementation). We supplement the baseline specification with controls for fee conditions and market intensity, respectively: an L2 base fee measure in Ether (\textsf{BaseGasL2}) and within-chain DEX volume in log USD as a DeFi activity proxy (\textsf{VolumeDEX}).

All outcome variables and time-varying covariates (except for the binary treatment indicator) are standardized (z-scored) prior to estimation. This choice is practical and methodological: chains differ by orders of magnitude in baseline transaction counts, gas usage, and revenue, and raw-unit levels can cause the regression to be dominated by the largest networks, conflating scale with treatment response. Standardization produces estimates in standard-deviation units, improving comparability across chains and outcomes and making effect sizes interpretable on a common scale. Importantly, standardization does not change the sign or statistical significance patterns relative to monotone transformations.

\paragraph{Empirical Specification and Identification.}
Our specification is given by the following model:
\begin{gather*}
    Y_{c,t} = \beta\cdot\textsf{Post$\times$Treated}_{c,t} +\gamma^\top X_{c,t} + \alpha_c+\delta_t + \varepsilon_{c,t}
\end{gather*}
where the dependent variable of interest (repeated transactions or revenue) is
\begin{gather*}
    Y\in\{\log\textsf{RepTXs},\log\textsf{Revenue}\},
\end{gather*}
$\textsf{Post$\times$Treated}_{c,t}=1_{\{c=\textsf{Arbitrum},\ t\ge\textsf{2025-04-17}\}}$ denotes the treatment effect (i.e.\ the differential change on Arbitrum after Timeboost relative to the change on control chains over the same period),
\begin{gather*}
    X\subseteq\{\textsf{BaseGasL2},\textsf{VolumeDEX}\}
\end{gather*}
is a set of controls, and $\varepsilon_{c,t}$ is idiosyncratic noise. We include chain fixed effects $\alpha_c$ that absorb time-invariant differences across networks (e.g., average activity levels, application mix, sequencer infrastructure) and day fixed effects $\delta_t$ that absorb common shocks shared across chains (e.g., market-wide volatility, macro crypto news, broad changes in user demand).

A key identification assumption is parallel trends conditional on fixed effects and controls: absent Timeboost, would Arbitrum’s outcomes have behaved similarly to the control L2s, after conditioning on fixed effects and control variables? This assumption is plausible because all networks share the same EVM execution environment and no other contemporaneous changes to ordering rules occurred on the control chains during the sample period. The outcomes we study arise from latency-sensitive MEV opportunities exist on all rollups, with ordering ultimately depending on arrival time even under priority-fee mechanisms. Hence, traders face the same competitive environment across chains. Chain fixed effects absorb persistent differences in scale and application composition, while day fixed effects absorb market-wide shocks.

\begin{table}[!htbp] \centering \footnotesize
\begin{tabular}{@{\extracolsep{2pt}}lcccc}
\\[-2.3ex]\hline
\hline \\[-2.3ex]
\\[-2.3ex] & (1) & (2) & (3) & (4) \\
\\[-2.3ex] & log\,\textsf{RepTXs} & log\,\textsf{RepTXs} & log\,\textsf{Revenue} & log\,\textsf{Revenue} \\
\hline \\[-2.3ex]
\textsf{Post$\times$Treated} & -0.699$^{**}$ & -0.698$^{***}$ & 0.726$^{***}$ & 0.711$^{***}$ \\
& (0.281) & (0.254) & (0.123) & (0.112) \\
\textsf{BaseGasL2} & & 0.096$^{}$ & & 0.150$^{**}$ \\
& & (0.076) & & (0.071) \\
\textsf{VolumeDEX} & & 0.159$^{}$ & & 0.078$^{}$ \\
& & (0.125) & & (0.150) \\
\hline \\[-2.3ex]
 Observations & 605 & 605 & 605 & 605 \\
 N. of groups & 5 & 5 & 5 & 5 \\
 $R^2$ & 0.025 & 0.041 & 0.001 & 0.032 \\
\hline
\hline \\[-2.3ex]
\multicolumn{5}{l}{
Includes chain and day fixed effects. Standard errors are clustered by chain.}  \\
\textit{Note:} & \multicolumn{4}{r}{$^{*}$p$<$0.1; $^{**}$p$<$0.05; $^{***}$p$<$0.01} \\
\end{tabular}
\caption{This table reports two-way fixed-effects estimates of the impact of Timeboost adoption on Arbitrum relative to control Layer-2 blockchains. The sample period spans February 17 to June 17, 2025, with Timeboost implemented on April 17, 2025. The dependent variables are the natural logarithm of repeated transactions (columns 1-2) and the natural logarithm of platform revenue in ETH (columns 3-4). Repeated transactions are defined as bursts of identical transactions (same sender, recipient, value, function selector, and calldata) within 2-second windows. Platform revenue includes gas fees from reverted transactions and, for Arbitrum post-Timeboost, auction payments from express lane winners. Post×Treated is an indicator equal to one for Arbitrum observations after April 17, 2025. BaseGasL2 is the Layer-2 base gas fee in ETH. VolumeDEX is the natural logarithm of DEX trading volume in USD.}
\label{tab:did-result-main}
\end{table}

\vspace{-1em}

\paragraph{Discussion of Results.}

Table \ref{tab:did-result-main} reports two-way fixed-effects estimates of the effect of Timeboost adoption on repeated transaction activity and sequencer revenue. The coefficient on \textsf{Post$\times$Treated} is negative and statistically significant (about $-0.70$ standard deviations), indicating that after Timeboost implementation, Arbitrum experienced a sizable decline in repeated transactions relative to the comparison L2s. This estimate is nearly unchanged when adding controls, suggesting that this reduction is not merely driven by coincident changes in transaction costs or market-wide trading activity. 

The \textsf{PostTreated} coefficient for revenue is positive and significant, implying that after Timeboost, revenue from reverted transactions and gas fees increased overall.  In the appendix, we provide additional robustness checks: we vary the window of measuring repeated transactions from 2 seconds to 5 seconds, and change the duration of the window before and after the Timeboost adoption date. These results carry through, with the decline in repeated transactions and gain in revenue after Timeboost adoption being significant across these changes. Additionally, as clustered standard errors may introduce  bias when the number of groups is small, we also quantify the uncertainty of our estimated coefficients using heteroskedasticity-robust standard errors (``HC3''), find that the empirical results hold under this as well. 


\subsection{Within-Arbitrum Analysis}

Our cross-chain empirical design controlled variation across chains and time by introducing fixed effects.
Now, we complement the analysis with a within-Arbitrum analysis by directly incorporating variables that would potentially influence transaction volumes.
This exercise does not replace the TWFE identification strategy, as it lacks an external control group, but provides a useful additional robustness check: if Timeboost meaningfully altered submission behavior, we should observe significant changes on Arbitrum itself, particularly after conditioning on proxies for activity (DEX volume) and transaction costs (L1 and L2 base fees).

We estimate a model of the form
\begin{equation*}
Y_t = \alpha + \beta\cdot\textsf{Post}_t + \gamma^\top X_t + \varepsilon_t
\end{equation*}
where the dependent variable of interest is
\begin{gather*}
    Y\in\{\log\textsf{RepTXs},\log\textsf{Revenue}\},
\end{gather*}
$\textsf{Post}_t=1_{\{t\geq \textsf{2025-04-17}\}}$ is a post-implementation indicator, and
\begin{gather*}
    X\subseteq\{\textsf{VolumeDEX},\textsf{BaseGasMain},\textsf{BaseGasArb},\textsf{VolatilityETH}\}
\end{gather*}
is a vector of time-varying controls. In all specifications we include a measure of DEX volume and base gas fees on Ethereum; we additionally include Arbitrum base gas fee and ETH volatility to capture additional shifts in arbitrage intensity. As a robustness check (reported in the appendix), we also verify that conclusions are not sensitive to alternative bandwidths around the implementation date and alternative time intervals for short-window repetition. Here, we do \textit{not} standardize variables before estimation.

\paragraph{Results.}
Table \ref{tab:the-result-main} reports the within-Arbitrum estimates. The post-Timeboost indicator is negative and statistically significant for \textsf{RepTXs}, implying a sizable decrease in repeated submissions after the implementation date even after controlling for DEX volume and fee conditions; the estimated coefficients remains similar when adding mainnet base fees and market controls. Similarly to the cross-chain analysis, there is a positive significant relation between the post-Timeboost period and revenue, and the post-Timeboost period is associated with higher failure rates among non-timeboosted transactions. 

\vspace{1em}

\begin{table}[!htbp]
\centering
\setlength{\tabcolsep}{4pt} 
\renewcommand{\arraystretch}{1.15}
\begin{tabular}{@{\extracolsep{2pt}}lcccc}
\hline
\hline
& (1) & (2) & (3) & (4) \\
 & log\,\textsf{RepTXs} & log\,\textsf{RepTXs} & log\,\textsf{Revenue} & log\,\textsf{Revenue} \\
\hline
\textsf{Intercept} & -9.423 & -8.955 & 2.621 & -4.377 \\
& (5.991) & (6.110) & (7.327) & (7.449) \\
\textsf{Post} & -0.578$^{**}$ & -0.453$^{*}$ & 0.858$^{**}$ & 0.594$^{*}$ \\
& (0.238) & (0.240) & (0.354) & (0.360) \\
\textsf{VolumeDEX} & 1.004$^{***}$ & 0.882$^{***}$ & -0.026 & 0.427 \\
& (0.299) & (0.312) & (0.364) & (0.382) \\
\textsf{BaseGasMain} & -0.050 & -0.146 & 0.631$^{***}$ & 0.690$^{***}$ \\
& (0.158) & (0.155) & (0.189) & (0.198) \\
\textsf{BaseGasArb} & & 187.350$^{***}$ & & -94.886$^{*}$ \\
& & (54.911) & & (50.138) \\
\textsf{VolatilityETH} & & 0.262 & & -3.216$^{***}$ \\
& & (0.768) & & (0.846) \\
\hline
Observations & 121 & 121 & 121 & 121 \\
$R^2$ & 0.173 & 0.257 & 0.068 & 0.141 \\
Adjusted $R^2$ & 0.152 & 0.218 & 0.044 & 0.096 \\
\hline
\hline
\multicolumn{5}{l}{Standard errors are computed using the Newey-West estimator.} \\
\textit{Note:} & \multicolumn{4}{r}{$^{*}$p$<$0.1; $^{**}$p$<$0.05; $^{***}$p$<$0.01} \\
\end{tabular}
\caption{This table reports estimates from a within-Arbitrum analysis of the impact of Timeboost adoption on redundant transactions and platform revenue. The sample consists of daily observations on Arbitrum from February 17 to June 17, 2025, with Timeboost implemented on April 17, 2025. The dependent variables are the natural logarithm of repeated transactions (columns 1-2) and the natural logarithm of platform revenue in ETH (columns 3-4). Repeated transactions are defined as bursts of identical transactions (same sender, recipient, value, function selector, and calldata) within 2-second windows. Platform revenue includes gas fees from reverted transactions and, post-Timeboost, auction payments from express lane winners. Post is an indicator equal to one for observations after April 17, 2025. VolumeDEX is the natural logarithm of DEX trading volume in USD on Arbitrum. BaseGasMain is the Ethereum mainnet base gas fee in ETH. BaseGasArb is the Arbitrum base gas fee in ETH. VolatilityETH is the daily return volatility of ETH/USD.}
\label{tab:the-result-main}
\end{table}

DEX volume is positively related to repeated transactions, consistent with higher trading activity increasing the number of contested opportunities and associated transaction flow. The coefficient on \textsf{BaseGasArb} is large and positive in the repeated transactions regressions, which is expected as Arbitrum gas fees are determined by congestion and would be higher under more spammed transactions. The coefficients on \textsf{BaseGasMain} and \textsf{BaseGasArb} for the revenue regressions tell a more interesting story: revenue from reverted gas fees and auction proceeds appears to be driven by base gas on Ethereum, with Arbitrum gas fees and volatility having a second-order effect. Overall, the within-chain estimates provide complementary evidence that Timeboost's implementation is associated with a reduction in repeated submissions on Arbitrum.

\section{Conclusion}

This paper studies how explicit markets for time priority reshape competition for time-sensitive trading opportunities. Our theoretical framework isolates a simple mechanism: when access to a time advantage is allocated via a market—rather than via a latency race in a public queue—equilibrium incentives to invest in speed shift from wasteful effort (duplicate submissions or latency technology) to auction bids. In our general model, a single-winner time-priority mechanism reduces aggregate latency investment and reallocates surplus from dissipation into payments for priority. In the blockchain specialization, where latency effort manifests as duplicate on-chain transactions, this translates into less spam, lower revert-driven resource waste, and higher platform-level revenue from time-priority sales.

We bring this framework to data by exploiting the introduction of Timeboost, an express-lane time-priority auction on Arbitrum, and combining cross-chain and within-chain empirical designs. The cross-chain difference-in-differences evidence shows a large and robust post-adoption decline in repeated transaction activity on Arbitrum relative to comparable Layer-2 networks, consistent with reduced redundant competition in the public queue. At the same time, we document an increase in protocol revenue from gas fees and auction proceeds, indicating that a larger share of MEV-related surplus is captured through explicit payments rather than revert-driven dissipation. A within-Arbitrum analysis yields qualitatively similar patterns, reinforcing the interpretation that the introduction of an explicit time-priority mechanism causally reshapes the structure of competition for MEV. Overall, our theoretical and empirical findings provide novel insights into the design of exchanges in the world of high-frequency trading, presenting another approach to mitigating externalities that arise from speed technology competition while improving revenue implications for the platforms themselves.

\newpage

\bibliographystyle{aer}
\bibliography{timeboost}

\newpage

\renewcommand{\thetable}{A.\arabic{table}}
\setcounter{table}{0}

\appendix

\section{Proofs of Results}

\subsection{Proof of Proposition \ref{prop:eq-baseline}}

Let $K_{-i}=\sum_{-i}k_{-i}$. Differentiating, we have, for $k_i>0$
\[
\frac{\partial u_i}{\partial k_i}(k_i;k_{-i}) 
= V\,\frac{K_{-i}}{(k_i + K_{-i})^2} - C'(k_i),
\]
\[
\frac{\partial^2 u_i}{\partial k_i^2}(k_i;k_{-i})
= -2V\,\frac{K_{-i}}{(k_i + K_{-i})^3} - C''(k_i) < 0
\]
since $C''>0$ and $K_{-i}\ge 0$. Thus for each fixed $K_{-i}\ge 0$, $u_i(\cdot;K)$ is strictly concave on $(0,\infty)$, so player $i$ has at most one interior best response. A strategy profile $k_i=k$ for all $i\in[N]$ is an interior symmetric equilibrium iff
\begin{gather*}
    \frac{\partial u_i}{\partial k_i}(k;(N-1)k) = 0,
\end{gather*}
which is exactly~
\begin{align*}
f(k) &:= V\cdot\frac{(N-1)k}{(NK)^2} - C'(k) \\[0.25\baselineskip]
       &= V\cdot\frac{N-1}{N^2}\cdot\frac{1}{k} - C'(k) = 0.
\end{align*}
We now show $f$ has a unique positive root. Differentiating,
\[
f'(k) = -V\cdot\frac{N-1}{N^2}\cdot\frac{1}{k^2} - C''(k) < 0\quad\text{for all }k>0,
\]
so $f$ is strictly decreasing on $(0,\infty)$.  As $k\downarrow 0$, the term $(1/k)$ dominates while $C'(k)$ remains finite (by continuity of $C'$), so $\lim_{k\downarrow 0} f(k) = +\infty$. As $k\to\infty$, strict convexity and $C'(k)$ increasing imply $C'(k)\to\infty$, while the first term $O(1/k)\to 0$, hence $\lim_{k\to\infty} f(k) = -\infty$. By continuity and the intermediate value theorem, there exists at least one $k^*>0$ with $f(k^*)=0$. By strict monotonicity, this root is unique. At that $k^*$, strict concavity of $u_i(\cdot;K^*)$  where $K^*=Nk^*$ implies that $k^*$ is the unique best response to $k^*$, so $(k^*,\dots,k^*)$ is the unique symmetric equilibrium.

\subsection{Proof of Proposition \ref{prop:eq-submission}}





For a given loser $j$, let $K_{-l}=\sum_{-j}k_l^{(-j)}$. The loser’s payoff is
\[
u_j(k_l^{(j)};k_w,k_l^{(-j)}) = V e^{-k_w \lambda T}\,\frac{k_l^{(j)}}{k_l^{(j)}+k_w+K_{-l}} - C(k_l^{(j)}).
\]
Differentiating, we have, for $k_l^{(j)}>0$,
\[
u_j'(k_l^{(j)};k_w,k_l^{(-j)})
 = V e^{-k_w \lambda T}\,\frac{K_{-l}}{(k_l^{(j)}+k_w+K_{-l})^2} - C'(k_l^{(j)}),
\]
\[
u_j''(k_l^{(j)};k_w,k_l^{(-j)})
 = -2V e^{-k_w \lambda T}\,\frac{K_{-l}}{(k_l^{(j)}+k_w+K_{-l})^3} - C''(k_l^{(j)}) < 0.
\]
Thus $u_j(\cdot)$ is strictly concave on $(0,\infty)$ when at least two losers are present. Using an argument analogous to that in the proof of Proposition~\ref{prop:eq-baseline}, the derivative $u_j'(k)$ is continuous, strictly decreasing in $k$, tends to $+\infty$ as $k\downarrow 0$ and to $-\infty$ as $k\to\infty$, so for each $(k_w,B,\lambda T)$ there is a unique interior maximizer $k>0$ solving $u_j'(k)=0$. 

From the winner's persepective, let $K_l=\sum_jk_l^{(j)}$. The winner's payoff satisfies
\begin{gather*}
u'_w(k_w;k_l^{(j)})
= V K_l e^{-k_w \lambda T}\frac{\lambda T(k_w+K_l) + 1}{(k_w+K_l)^2} - C'(k_w) \\[0.25\baselineskip]
u''_w(k_w;k_l^{(j)})
= -2VKe^{-k_w \lambda T}\frac{\lambda T (k_w+K_l) +1}{(k_w+K_l)^3}
 - V K e^{-k_w \lambda T}\frac{\lambda T^2 (k_w+K_l)^2}{(k_w+K_l)^4} - C''(k_w) < 0
\end{gather*}
for $k_w>0$ and $K_l>0$. Similarly $u'_w$ is continuous, strictly decreasing in $k_w$, with limit $+\infty$ as $k_w\downarrow 0$ and $-\infty$ as $k_w\to\infty$, so there is a unique interior best response $k_w>0$. For a symmetric loser profile $k=k_\ell$, set $K_{-l}=(N-2)k_\ell$ and $K_l=(N-1)k_\ell$. The first-order conditions are then
\begin{gather*}
C'(k_w) 
=
V\left(\frac{(N-1)k_l}{((N-1)k_l+k_w)^2}+\frac{(N-1)k_l}{(N-1)k_l+k_w}\,\lambda T\right)\exp(-k_w\lambda T),
\\[0.25\baselineskip]
C'(k_l) 
=
V\left(\frac{(N-2)k_l+k_w}{((N-1)k_l+k_w)^2}\right)\exp(-k_w\lambda T).
\end{gather*}
Define the right-hand sides as $R_w(k_l,k_w)$ and $R_l(k_l,k_w)$.

We first claim that for each fixed $k_w>0$, at most one $k_l>0$ solves $C'(k_l)=R_l(k_l,k_w)$. For a fixed $k_w>0$, note that
\begin{gather*}
    \frac{d}{dk_l}\left(\frac{(N-2)k_l+k_w}{((N-1)k_l+k_w^2}\right) = \frac{-(N-1)(N-2)k_l-Nk_w}{((N-1)k_l+k_w)^2}<0
\end{gather*}
so $R_l(k_w,k_w)$ is decreasing in $k_l$. Since $C'(k_l)$ is strictly increasing in $k_l$ by convexity of $C$. Therefore the equation $C'(k_l)=R_l(k_l,k_w)$ has at most one solution $k_l>0$ for each $k_w>0$. Where a solution exists, denote it by $k_l=\kappa_l(k_w)$. We then show that $\kappa_l(k_w)$ is strictly decreasing in $k_w$. For fixed $k_l = 0$, note that
\begin{gather*}
    \frac{d}{dk_w}\left(\frac{(N-2)k_l+k_w}{\bigl((N-1)k_l+k_w\bigr)^2}\right)
=\frac{-(N-3)k_l-k_w}{((N-1)k_l+k_w)^3}<0.
\end{gather*}
Thus, $R(k_l,k_w)$ is strictly decreasing in $k_w$ for each fixed $k_l>0$. Along the solution curve $C'(\kappa_l(k_w))=R_l(\kappa_l(k_w),k_w)$, if $k_w$ increases, the RHS decreases. Since $C'$ is strictly increasing, $\kappa_l(k_w)$ must strictly decrease to maintain equality. Hence, $\kappa_l(k_w)$ is strictly decreasing.

We then show that $k_w\mapsto R_w(\kappa_l(k_w),k_w)$ is strictly decreasing. Define
\begin{gather*}
    G(k_w):=C'(k_w)-\rho_w(\kappa_l(k_w),k_w).  
\end{gather*}

Because $C'$ is strictly increasing, $k_w\mapsto C'(k_w)$ is strictly increasing. Let $K=(N-1)k_l+k_w$ and write
\[
R_w(k_l,k_w)=V(N-1)\,k_l\left(\frac{1}{K^2}+\frac{\lambda T}{K}\right)e^{-k_w\lambda T},
\]
Along $k_l=\kappa_l(k_w)$ we have: (i) $\kappa_l$ is strictly decreasing, so $k_l$ decreases; (ii) $K$ increases in $k_w$ (since $k_w$ increases and $k_l$ decreases only weakly), and $\frac{1}{K^2}+\frac{\lambda T}{K}$ is strictly decreasing in $K$; (iii) $e^{-k_w\lambda T}$ is strictly decreasing in $k_w$. Thus $R_w(\kappa_l(k_w),k_w)$ is strictly decreasing in $k_w$. Therefore $G(k_w)$ is strictly increasing, so $G(k_w)=0$ has at most one solution $k_w^*$. Then $k_l^*=\kappa_l(k_w^*)$ is uniquely determined.

We finally show that if a solution to the above system exists, then $k_l^*<k_w^*$. Let $R_w^*=C'(k_w^*)$, $R_l^*=C'(k_l^*)$, and $K^* = (N-1)k_l^*+k_w^*$. Then
\begin{align*}
R_w^*-\rho_l^*
&=V e^{-k_w^*\lambda T}
\left(
\frac{(N-1)k_l^*-((N-2)k_l^*+k_w^*)}{(K^*)^2}
+\frac{(N-1)k_l^*}{K^*}\lambda T
\right)\\
&=V e^{-k_w^*\lambda T}
\left(
\frac{k_l^*-k_w^*}{(K^*)^2}
+\frac{(N-1)k_l^*}{K^*}\lambda T
\right).
\end{align*}
Assume, towards a contradiction, that $k_w^*\le k_l^*$. Then $k_l^*-k_w^*\ge 0$ and the second term is strictly positive because $\lambda T>0$, $N-1>0$, $k_l^*>0$, and $K^*>0$. Hence $R_w^*-\rho_l^*>0$, i.e.\ $R_w^*>R_l^*$. Using the FOCs this gives $C'(k_w^*)>C'(k_l^*)$. But $C$ convex implies $C'$ is nondecreasing, so $k_w^*\le k_l^*$ would imply $C'(k_w^*)\le C'(k_l^*)$, a contradiction. Thus $k_l^*<k_w^*$.

We finally show that $0<u_l^*<u_w^*$. Write
\begin{gather*}
    u_l(k_l^*;k_w) = V\cdot\frac{k_l^*}{K(k_w)}\cdot e^{-k_w\lambda T}-C(k_l^*)
\end{gather*}
where $K(k_w)=(N-1)k_l^*+k_w$. Then $\partial u_l/\partial k_w<0$, so because $k_w^*>k_l^*$, so we have
\begin{gather*}
    u_l^* = u_l(k_l^*;k_w^*) < u_l(k_l^*;k_l^*)
\end{gather*}
From strict concavity of $u_w(\cdot;k_l^*)$ and the FOC, $k_w^*$ is the unique maximizer (assuming an equilibrium exists), so
\begin{gather*}
    u_w^* = u_w(k_w^*;k_l^*) \ge u_w(k_l^*;k_l^*)
\end{gather*}
At $k_w=k_l^*$ we have $K(k_l^*)=Nk_l^*$. Therefore
\[
u_w(k_l^*;k_l^*)
=
V\!\left(1-\frac{N-1}{N}e^{-k_l^*\lambda T}\right)-C(k_l^*),
\]
\[
u_l(k_l^*;k_l^*)
=
V\!\left(\frac{1}{N}e^{-k_l^*\lambda T}\right)-C(k_l^*).
\]
Subtracting the second equation from the first, we have,
\[
u_w(k_l^*;k_l^*)-u_l(k_l^*;k_l^*)
=
V\left(1-e^{-k_l^*\lambda T}\right)>0
\]
Combining the chain of inequalities yields
\[
u_w^*-u_l^*
\ge
u_w(k_l^*;k_l^*)-u_l(k_l^*;k_l^*)
=
V\bigl(1-e^{-k_l^*\lambda T}\bigr)
>0.
\]


\subsection{Proof of Proposition \ref{prop:eq-auction}}

Let $N\ge 3$ bidders participate in a sealed-bid second-price auction and
\[
v_A \;:=\; u_w^* - u_l^*
\]
denote each bidder's (common, known) incremental value of winning the Express Lane right,
i.e.\ the difference between being the winner and being a loser in the continuation game.
Conditional on the submission-subgame equilibrium being played after the auction,
a bidder who wins and pays price $p$ obtains utility $u_w^* - p$, while a bidder who loses
obtains utility $u_l^*$. Thus the incremental gain from winning at price $p$ is $v_A-p$.

Consider a symmetric pure bid profile in which all bidders submit the same bid $b\ge 0$.
Under uniform tie-breaking, each bidder wins with probability $1/N$ and pays the
(second-highest) price $b$ whenever they win. Hence each bidder's expected utility is
\begin{equation*}
U(b)
=
\frac{1}{N}(u_w^* - b) + \frac{N-1}{N}u_l^*
=
u_l^* + \frac{1}{N}(v_A-b).
\label{eq:U_of_b}
\end{equation*}
Fix such a symmetric profile and consider a unilateral deviation by bidder $i$:
\begin{enumerate}
    \item If $b<v_A$, bidder $i$ can deviate to $b+\varepsilon$ for any sufficiently small $\varepsilon>0$.
Then bidder $i$ wins with probability $1$ and pays the second-highest bid, which remains $b$.
The deviation payoff is therefore $u_w^* - b = u_l^* + (v_A-b)$, which strictly exceeds
$U(b)=u_l^*+(v_A-b)/N$ because $v_A-b>0$ and $N\ge 3$. Hence no symmetric profile with $b<v_A$
can be a Nash equilibrium.
    \item If $b>v_A$, bidder $i$ can deviate to a sufficiently low bid (e.g.\ $0$), guaranteeing that
they lose and obtain $u_l^*$. Under the symmetric profile, their payoff is $U(b)=u_l^*+(v_A-b)/N<u_l^*$.
Thus any symmetric profile with $b>v_A$ admits a profitable deviation and cannot be a Nash equilibrium.
    \item If $b=v_A$, then $U(v_A)=u_l^*$. Deviating to any bid $b_i<v_A$ guarantees losing and yields $u_l^*$,
which equals $U(v_A)$. Deviating to any bid $b_i>v_A$ guarantees winning and paying $v_A$, yielding
$u_w^*-v_A=u_l^*$, again equal to $U(v_A)$. Therefore no deviation yields a strictly higher payoff,
so $b=v_A$ is a best response to the symmetric profile.
\end{enumerate}
Combining (1)--(3), it follows that the only symmetric pure bid profile that can be a Nash equilibrium is $b^*=v_A$, and it is indeed a Nash equilibrium. We proceed to verify sequential rationality in every subgame.
\begin{enumerate}
    \item From Proposition \ref{prop:eq-submission}, after any realized auction winner, the continuation strategies form a Nash equilibrium
of the corresponding submission subgame. Hence no player can profitably deviate after observing
the auction outcome.
    \item Given that the continuation equilibrium is played after the auction, each bidder's value for winning
the Express Lane right is $v_A=u_w^*-u_l^*$. By the previous proposition, bidding $b^*=v_A$
is a best response (indeed, the unique symmetric pure-strategy equilibrium bid) to the others bidding $b^*$.
Therefore no bidder can profitably deviate in the auction stage.
\end{enumerate}
Thus, the overall strategy profile is subgame-perfect.

\subsection{Proof of Theorem \ref{thm:comparison-general}}

Multiply the loser's FOC by $(N-1)$ and add the winner's FOC. Using $K^*=(N-1)k_l^*+k_w^*$, we have
\begin{align}
(N-1)\cdot C'(k_l^*)+C'(k_w^*)
&=
V e^{-k_w^*\lambda T}\Bigg[
(N-1)\frac{(N-2)k_l^*+k_w^*}{(K^*)^2}
+\frac{(N-1)k_l^*}{(K^*)^2}
+\frac{(N-1)k_l^*\lambda T}{K^*}
\Bigg]\notag\\
&=
V e^{-k_w^*\lambda T}\left[
\frac{(N-1)\big((N-2)k_l^*+k_w^*+k_l^*\big)}{(K^*)^2}
+\frac{(N-1)k_l^*\lambda T}{K^*}
\right]\notag\\
&=
V e^{-k_w^*\lambda T}\left[
\frac{(N-1)D}{(K^*)^2}
+\frac{(N-1)k_l^*\lambda T}{K^*}
\right]\notag\\
&=
\frac{V(N-1)}{K^*}\,e^{-k_w^*\lambda T}\,(1+k_l^*\lambda T).
\label{eq:sumIdentity}
\end{align}
Let $\theta:=1/N$. Since $\frac{K^*}{N}=\theta k_w^*+(1-\theta)k_l^*$ and $\log C'$ is convex,
\begin{equation*}\label{eq:logconvex}
\log C'\!\left(\frac{K^*}{N}\right)
\le \theta\log C'(k_w^*)+(1-\theta)\log C'(k_l^*),
\end{equation*}
so exponentiating gives
\begin{equation}\label{eq:geomMeanBound}
C'\!\left(\frac{K^*}{N}\right)
\le (C'(k_w^*))^{1/N}(C'(k_l^*))^{(N-1)/N}.
\end{equation}
By weighted AM--GM,
\begin{equation}\label{eq:amgm}
(C'(k_w^*))^{1/N}(C'(k_l^*))^{(N-1)/N}
\le \frac{C'(k_w^*)+(N-1)\cdot C'(k_l^*)}{n}.
\end{equation}
Combining \eqref{eq:geomMeanBound}--\eqref{eq:amgm} with \eqref{eq:sumIdentity} yields
\begin{equation}\label{eq:CprimeDoverN}
C'\!\left(\frac{K^*}{N}\right)
\le \frac{1}{n}\cdot \frac{V(N-1)}{K^*}\,e^{-k_w^*\lambda T}\,(1+k_l^*\lambda T).
\end{equation}
For all $t\ge 0$, $(1+t)e^{-t}\le 1$, with strict inequality for $t>0$.
Let $t:=k_w^*\lambda T>0$. Since $k_l^*<k_w^*$, we have $1+k_l^*\lambda T\le 1+ k_w^*\lambda T=1+t$.
Thus
\begin{equation}\label{eq:strictFactor}
e^{-k_w^*\lambda T}(1+k_l^*\lambda T)
\le e^{-t}(1+t) < 1.
\end{equation}
Substituting \eqref{eq:strictFactor} into \eqref{eq:CprimeDoverN} gives
\begin{gather*}\label{eq:keyStrict}
C'\!\left(\frac{K^*}{N}\right)
<
\frac{V(N-1)}{NK^*}.
\end{gather*}
Define, for $x>0$,
\begin{gather*}
F(x):=C'\!\left(\frac{x}{N}\right)-\frac{V(N-1)}{Nx}.
\end{gather*}
Because $C'$ is increasing and $x\mapsto -\frac{V(N-1)}{Nx}$ is increasing on $(0,\infty)$,
the function $F$ is increasing. By the baseline FOC, we have $F(K^*_{\mathrm{baseline}})=0$.
By \eqref{eq:keyStrict}, $F(K)<0$. Since $F$ is increasing, $F(K^*)<F(K^*_{\mathrm{baseline}})$ so $K^*<K^*_{\mathrm{baseline}}$.

\subsection{Proof of Proposition \ref{prop:comparison}}

\paragraph{Less Redundant Submissions.}

We show that when $C'(k)=C$, noting the abuse of notation, a solution to the equilibrium system always exists. The FOCs are now
\begin{gather*}
C = V\left(\frac{(N-1)k_l}{K^2}+\frac{(N-1)k_l}{K}\lambda T\right)e^{-k_w\lambda T},
\label{eq:lin1}\\[0.25\baselineskip]
C = V\left(\frac{(N-2)k_l+k_w}{K^2}\right)e^{-k_w\lambda T},
\label{eq:lin2}
\end{gather*}
where $K:=(N-1)k_l+k_w$.

For $k_l,k_w>0$ both right-hand sides are positive, so dividing the winner's FOC by the loser's FOC cancels $c$, $V$, and $e^{-k_w\lambda T}$ and yields
\[
1=\frac{\frac{(N-1)k_l}{K^2}+\frac{(N-1)k_l}{K}\lambda T}{\frac{(N-2)k_l+k_w}{K^2}}
=\frac{(N-1)k_l+(N-1)k_lK\lambda T}{(N-2)k_l+k_w}.
\]
Rearranging gives
\[
(N-1)k_l+(N-1)k_lK\lambda T=(N-2)k_l+k_w
\implies 
k_w=k_l+(N-1)\lambda T\,k_lK.
\]
Substituting $K=(N-1)k_l+k_w$ and solving for $k_w$ in terms of $k_l$ yields
\begin{equation}\label{eq:kw_of_kl}
\kappa_w(k_l)=\frac{k_l\bigl(1+(N-1)^2\lambda T\,k_l\bigr)}{1-(N-1)\lambda T\,k_l},
\qquad
k_l\in\Bigl(0,\frac{1}{(N-1)\lambda T}\Bigr),
\end{equation}
which is strictly positive on the stated interval. Along the curve
$\kappa_w(k_l)$, the two right-hand sides of the FOCs coincide. Define, for $k_l\in\bigl(0,\frac{1}{(N-1)\lambda T}\bigr)$,
\[
F(k_l):=
V\left(\frac{(N-2)k_l+\kappa_w(k_l)}{\bigl((N-1)k_l+\kappa_w(k_l)\bigr)^2}\right)
\exp\!\bigl(-\lambda T\,\kappa_w(k_l)\bigr),
\]
so that the loser's FOC is equivalent to $F(k_l)=c$ once $k_w=\kappa_w(k_l)$.

We claim $\lim_{k_l\downarrow 0}F(k_l)=+\infty$ and
$\lim_{k_l\uparrow \frac{1}{(N-1)\lambda T}}F(k_l)=0$.
First, from \eqref{eq:kw_of_kl} one has $\kappa_w(k_l)\sim k_l$ as $k_l\downarrow 0$.
Hence $K=(N-1)k_l+\kappa_w(k_l)\sim N k_l$ and $(N-2)k_l+\kappa_w(k_l)\sim (N-1)k_l$, so
\[
\frac{(N-2)k_l+\kappa_w(k_l)}{(K^*)^2}
\sim
\frac{(N-1)k_l}{(N k_l)^2}
=\frac{N-1}{N^2}\cdot \frac{1}{k_l}\to +\infty,
\]
while $\exp(-\lambda T\,\kappa_w(k_l))\to 1$. Thus $F(k_l)\to+\infty$ as $k_l\downarrow 0$. Second, as $k_l\uparrow \frac{1}{(N-1)\lambda T}$, the denominator
$1-(N-1)\lambda Tk_l\downarrow 0$ in \eqref{eq:kw_of_kl}, so $\kappa_w(k_l)\to +\infty$,
which implies $\exp(-\lambda T\,\kappa_w(k_l))\to 0$.
Moreover the rational factor is $O(1/\kappa_w(k_l))$, hence remains bounded and tends to $0$.
Therefore $F(k_l)\to 0$ as $k_l\uparrow \frac{1}{(N-1)\lambda T}$.

Since $F$ is continuous on
$\bigl(0,\frac{1}{(N-1)\lambda T}\bigr)$ and takes values from $+\infty$ down to $0$,
the intermediate value theorem implies that for every $c>0$ there exists
$k_l^*\in\bigl(0,\frac{1}{(N-1)\lambda T}\bigr)$ such that $F(k_l^*)=c$.
Setting $k_w^*:=k_w(k_l^*)>0$ yields a pair $(k_l^*,k_w^*)\in(0,\infty)^2$
satisfying the system. Hence $(k_l^*,k_w^*)$ solves the original system. The result follows naturally since a linear cost function $C(k)=ck$ is weakly convex, increasing, and has a weakly log-convex first-derivative.

\paragraph{More Sequencer Revenue.}

From our derived formulas for sequencer renuve in both the baseline setting and under Timeboost, this reduces to showing $u_l^* < V/N^2$ for any $T>0$. Let $K^*:=k_w^*+(N-1)k_l^*$ denote total intensity in the Express Lane equilibrium.
Using the losers' FOC to substitute for $C$ in the loser
payoff yields the identity
\begin{equation}
u_l^*
=
e^{-k_w^*\lambda T}\,V\left(\frac{k_l^*}{K^*}\right)^2.
\label{eq:ul-identity}
\end{equation}
Moreover, since $k_w^*>k_l^*$ when $T>0$, hence
$K^*=k_w^*+(N-1)k_l^* > N k_l^*$ and therefore $(k_l^*/K^*)^2 < 1/N^2$. Since
$e^{-k_w^*\lambda T}<1$ for $T>0$ and $k_w^*>0$, \eqref{eq:ul-identity} yields
\[
u_l^*
=
e^{-k_w^*\lambda T}\,V\left(\frac{k_l^*}{K^*}\right)^2
<
V\cdot \frac{1}{N^2}.
\]
Thus $u_l^*<V/N^2$, which implies the result.

\subsection{Preliminaries for 2-Type Allocations}

Suppose that $n_l, n_w\geq1$. Let $K_w=k_w^{(i)}+\sum_{-i}k_w^{(-i)}$, $K_l=k_l^{(j)}+\sum_{-j}k_l^{(-j)}$, and $K=K_w+K_l$. Payoffs are then
\begin{gather*}
    u_w^{(i)}(k_w^{(i)};k_w^{(-i)},k_l^{(j)}) = V\cdot\frac{k_w^{(i)}}{K_w}\left(1-\frac{K_l}{K}\cdot\exp(-\lambda T K_w )\right)-Ck_w^{(i)}, \\[0.25\baselineskip]
    u_l^{(j)}(k_l^{(j)};k_w^{(i)},k_l^{(-j)}) = V\left(\frac{k_l^{(j)}}{K}\cdot\exp(-\lambda TK_w)\right)-Ck_l^{(j)}.
\end{gather*}
For fixed winner $i$ and loser $j$, let $K_{-w} = \sum_{-i}k_w^{(-i)}$ and $K_{-l}=\sum_{-j}k_l^{(-j)}$. Differentiating, we have
\begin{gather*}
    \frac{du_w^{(i)}}{dk_w^{(i)}} = V\left[\frac{K_{-w}}{K_w^2}\left(1-\frac{K_l}{K}e^{-\lambda T K_w}\right)+\frac{k_w^{(i)}}{K_w}\left(\frac{K_l}{K^2}+\lambda T\frac{K_l}{K}\right)e^{-\lambda T K_w}\right] - C \\[0.25\baselineskip]
    \frac{du_l^{(j)}}{dk_l^{(j)}} = V\cdot\frac{K_w+K_{-l}}{K^2}\cdot e^{-\lambda T K_w} - C
\end{gather*}
Differentiating the loser's condition again, noting that \(K_w\) is constant in $k_l^{(j)}$ and \(\frac{dK}{dk_l^{(j)}}=1\), we have
\[
\frac{d^2 u_l^{(j)}}{d (k_l^{(j)})^2}
=
V\,(K_w+K_{-l})\,e^{-\lambda T K_w}\,\frac{d}{dk_l^{(j)}}(K^{-2})
=
-\,\frac{2V\,(K_w+K_{-l})}{K^3}\,e^{-\lambda T K_w} \le 0 .
\]
For the winner's condition, we have
\begin{gather*}
    \frac{d^2 u_w^{(i)}}{d(k_w^{(i)})^2}<0.
\end{gather*}
The proof is given in the Additional Proofs section.

Imposing symmetry, i.e.\ $K_l=n_lk_l$, $K_{-l}=(n_l-1)k_l$, $K_w=n_wk_w$, and $K_{-w}=(n_w-1)k_w$, the FOCs are then
\begin{gather*}
C
=
V\left[
\frac{n_w-1}{n_w^2\,k_w}
\left(
1-\frac{n_l k_l}{n_w k_w+n_l k_l}\,
e^{-\lambda T n_w k_w}
\right)
+
\frac{1}{n_w}
\left(
\frac{n_l k_l}{(n_w k_w+n_l k_l)^2}
+
\lambda T\,\frac{n_l k_l}{n_w k_w+n_l k_l}
\right)
e^{-\lambda T n_w k_w}
\right], \\[0.25\baselineskip]
C
=
V\,
\frac{n_w k_w+(n_l-1)k_l}{(n_w k_w+n_l k_l)^2}\,
e^{-\lambda T n_w k_w}.
\end{gather*}
Equilibrium effort is $K^*=n_wk_w^*+n_lk_l^*$ and equilibrium continuation payoffs are then
\begin{gather*}
    u_w^* = V\cdot\frac{1}{n_w}\left(1-\frac{n_lk_l^*}{n_lk_l^*+n_wk_w^*}\cdot\exp(-\lambda T n_wk_w^*)\right)-Ck_w^* \\[0.25\baselineskip]
    u_l^* = V\cdot\frac{k_l^*}{n_lk_l^*+n_wk_w^*}\exp(-\lambda T n_wk_w^*)-Ck_l^*
\end{gather*}
Let $W=n_wk_w$, $L=n_lk_l$, $K=W+L$, $\alpha_l=1-n_l^{-1}$, $\alpha_w=1-n_w^{-1}$, and $\tau=\lambda T$. Assume that $V,C, \lambda, T > 0$ and $\alpha_l,\alpha_w\in[0,1)$. Rewrite the FOCs as
\begin{gather*}
    F_w(W,L, T):=
\alpha_w\Bigl(\frac{1}{W}\cdot\frac{W+L(1-e^{-\tau W})}{W+L}\Bigr)
+(1-\alpha_w)\Bigl(\frac{L}{(W+L)^2}+\frac{\tau L}{W+L}\Bigr)\exp(-\tau W)-\rho \\[0.25\baselineskip]
F_l(W,L, T):=\frac{W+\alpha_l L}{(W+L)^2}e^{-\tau W}-\rho, \\[0.25\baselineskip]
\alpha_l = \frac{(N-1)(1-\alpha_w)-1}{N(1-\alpha_w)-1} \\[0.25\baselineskip]
u_w^* = V\left[(1-\alpha_w)\left(1-\frac{L}{W+L}e^{-\tau W}\right)-\rho(1-\alpha_w)W\right] \\[0.25\baselineskip]
u_l^* = V\left[(1-\alpha_l)\cdot\frac{L}{W+L}e^{-\tau W}-\rho(1-\alpha_l)L\right]
\end{gather*}

\subsection{Proof of Proposition \ref{thm:fix-n-opt-t}}

\paragraph{Less Redundant Submissions.}

Let $K( T)=n_w k_w( T)+n_l k_l( T)$ denote total equilibrium effort as a function of the advantage difference $ T=T$, where $(k_w( T),k_l( T))$ solves the symmetric equilibrium first-order conditions. At $ T=0$, the game is symmetric across all $N=n_w+n_l$ players, and the unique interior equilibrium $k_0$ is given by the baseline case.

Differentiating the symmetric equilibrium first-order conditions with respect to $ T$ and evaluating at $ T=0$ yields
\begin{gather*}
    k_w'(0)=\lambda n_l (N-1)k_0^2,\\[0.25\baselineskip]
    k_l'(0)=-\lambda n_w (N-1)k_0^2.
\end{gather*}
so $K'(0) = n_w k_w'(0)+n_l k_l'(0) = 0$. Differentiating again and solving for $k_w''(0)$ and $k_l''(0)$ yields
\begin{gather*}
k''_w(0)=-\frac{V^3\lambda^2}{C^3}\frac{n_l\,(N-1)^2}{N^6}
\Bigl(2n_l^3+5n_l^2n_w-3n_l^2+4n_ln_w^2-3n_ln_w+n_l+n_w^3\Bigr), \\[0.25\baselineskip]
k''_l(0)=-\frac{V^3\lambda^2}{C^3}\frac{n_w\,(N-1)^2}{N^6}
\Bigl(n_l^2n_w+n_l^2+2n_ln_w^2-n_ln_w+n_w^3-2n_w^2+n_w\Bigr), \\[0.25\baselineskip]
K''(0) = n_w k_w''(0)+n_l k_l''(0)
= -\frac{V^3\lambda^2}{C^3}
\frac{n_w n_l (N-1)^2}{N^5}
\bigl(2N(N-1)+1\bigr) < 0
\end{gather*}
since $n_w,n_l\ge1$ and $\lambda>0$. By Taylor's theorem,
\[
K'( T)=K'(0)+K''(0) T+o( T)
=K''(0) T+o( T).
\]
Since $K''(0)<0$, there exists$\delta>0$ such that $K'( T)<0$ for all $ T\in(0,\delta)$. Thus, total equilibrium effort is locally decreasing in $ T$ to the right of $0$.

\paragraph{More Sequencer Revenue.} Define $L( T)=n_l k_l^*( T)$, $W( T)=n_w k_w^*( T)$, $K( T)=W( T)+L( T)$, and $E( T)=e^{-\lambda T L( T)}$. From earlier, we have
\[
W'(0)=-\lambda n_w n_l (N-1)k_0^2,
\quad
L'(0)=\lambda n_w n_l (N-1)k_0^2,
\quad
E'(0)=-\lambda B(0)=-\lambda n_w k_0
\]
Define $f( T)=\frac{W( T)}{K( T)}E( T)$; then for the winner's payoff:
\begin{gather*}
    u_w( T)=\frac{V}{n_w}\bigl(1-f( T)\bigr)-Ck_w^*( T), \\[0.5\baselineskip]
    u_w'( T)=-\frac{V}{n_w}f'( T)-Ck_w'( T).
\end{gather*}
Using $K'(0)=0$ and $E(0)=1$, we have
\[
f'(0)=\frac{W'(0)}{K(0)}+\frac{W(0)}{K(0)}\cdot E'(0)
=-\lambda n_w n_l k_0.
\]
so, substituting the expression for $k_0$, it follows that
\[
u_w'(0)
=
\lambda n_l\bigl(Vk_0-C(N-1)k_0^2\bigr)\propto\frac{V^2}{C}\frac{(N-1)(2N-1)}{N^4}>0.
\]
For the loser's payoff, define $g( T)=\frac{k_l^*( T)}{K( T)}E( T)$; then
\begin{gather*}
    u_l( T)=Vg( T)-Ck_l^*( T), \\[0.25\baselineskip]
    u_l'( T)=Vg'( T)-Ck_l'( T)
\end{gather*}
Again using $K'(0)=0$,
\begin{gather*}
    g'(0)=\frac{k_l'(0)}{K(0)}+\frac{k_l(0)}{K(0)}E'(0)
=-\lambda n_w k_0
\end{gather*}
Thus $u_l'(0)
=
\lambda n_w\bigl(-Vk_0+C(N-1)k_0^2\bigr)
=
-\lambda n_w\bigl(Vk_0-C(N-1)k_0^2\bigr)<0$. Auction revenue $n_w(u_w-u_l)$ and sequencer revenue $V-Nu_l$ are thus increasing in $ T$ around a right-neighborhood of $ T=0$ for any $n_w,n_l\geq1$.

\subsection{Proof of Proposition \ref{thm:fix-t-opt-n}}

\paragraph{Less Redundant Submissions.} Fix $\tau\coloneqq\lambda T>0$ (slight abuse of notation) and $\rho\coloneqq C/V>0$. For each $N\ge 2$, let $(W(\alpha_w;N),L(\alpha_w;N))$ be an equilibrium branch satisfying
\begin{gather*}
    F_w(W(\alpha_w;N),L(\alpha_w;N),\alpha_w)=0, \\[0.25\baselineskip]
    F_l(W(\alpha_w;N),L(\alpha_w;N),\alpha_w)=0,
\end{gather*}
Let $J$ denote the Jacobian of the equilibrium system, i.e.\
\begin{gather*}
    J=
    \begin{pmatrix}
    F_{wW} & F_{wL} \\
    F_{lW} & F_{lL} 
    \end{pmatrix}.
\end{gather*}
From the baseline model with time-priority auction, there exists $W,L>0$ satisfying both equations at $\alpha_w=0$ for any $N\geq2$. Since the system has at most one interior solution, it follows that $J$ is nonsingular at $\alpha_w=0$ for any $N\geq2$ as well. We claim that
\[
\lim_{N\to\infty}\bigl(W'(0)+L'(0)\bigr)<0,
\]
where $W'(0)=\frac{dW}{d\alpha_w}\big|_{\alpha_w=0}$ and $L'(0)=\frac{dL}{d\alpha_w}\big|_{\alpha_w=0}$.
Let $K:=W+L$ and $E:=e^{-\tau W}$. Define
\begin{gather*}
    A(W,L):=\frac{1}{W}\cdot\frac{W+L(1-E)}{K} \\[0.25\baselineskip]
    B(W,L):=\frac{L}{K^2}+\frac{\tau L}{K}
\end{gather*}
Then FOCs can be written as 
\begin{gather*}
    F_w(W,L,\alpha_w)=\alpha_w \cdot A(W,L)+(1-\alpha_w)\cdot B(W,L)\cdot E-\rho \\[0.25\baselineskip]
    F_l(W,L,\alpha_w)=\frac{W+\alpha_l L}{K^2}E-\rho
\end{gather*}
where $\alpha_l=\frac{(N-1)(1-\alpha_w)-1}{N(1-\alpha_w)-1}$. At $\alpha_w=0$, we denote evaluation by a superscript $0$. 

Differentiate the equilibrium conditions with respect to $\alpha_w$:
\begin{gather*}
    F_{wW}W'+F_{wL}L'+F_{w\alpha_w}=0, \\[0.25\baselineskip]
F_{lW}W'+F_{lL}L'+F_{l\alpha_w}=0.
\end{gather*}
Then at $\alpha_w=0$, by Cramer's rule,
\begin{gather*}
    W'(0)=\frac{F_{wL}^0F_{l\alpha_w}^0-F_{lL}^0F_{w\alpha_w}^0}{\det J^0}, \\[0.25\baselineskip]
L'(0)=\frac{F_{lW}^0F_{w\alpha_w}^0-F_{wW}^0F_{l\alpha_w}^0}{\det J^0}.
\end{gather*}
Summing yields
\begin{gather*}\label{eq:sum}
K'(0)=
\frac{(F_{lW}^0-F_{lL}^0)\,F_{w\alpha_w}^0+(F_{wL}^0-F_{wW}^0)\,F_{l\alpha_w}^0}{\det J^0}.
\end{gather*}
From $F_w=\alpha_w A+(1-\alpha_w)BE-\rho$, differentiating and evaluating at $\alpha_w=0$ yields $F_{w\alpha_w}^0=A-BE$. At $\alpha_w=0$, the equilibrium condition $F_w=0$ gives $BE=\rho$, hence
\begin{equation*}\label{eq:fwalpha}
F_{w\alpha_w}^0=A-\rho.
\end{equation*}
Next, $\alpha_w$ enters $F_l$ only through $\alpha_l$, so
\[
F_{l\alpha_w}=\frac{\partial F_l}{\partial \alpha_l}\cdot \frac{d\alpha_l}{d\alpha_w} = -\frac{L}{K^2}\cdot E\cdot\frac{1}{(N(1-\alpha_w)-1)^2}.
\]

At $\alpha_w=0$,
\begin{gather*}\label{eq:flalpha}
F_{l\alpha_w}^0
=
-\frac{L}{K^2}\frac{E}{(N-1)^2}
=
O\!\left(\frac{1}{(N-1)^2}\right).
\end{gather*}
At $\alpha_w=0$ we have $\alpha_l=\frac{N-2}{N-1}\to 1$ as $N\to\infty$.
For fixed $(W,L)$, the  formulas for the $(W,L)$-partials of $F_l$ are
\begin{gather*}
    F_{lW}^0=E\left[\frac{-W+(1-2\alpha_l)L}{K^3}-\tau\frac{W+\alpha_l L}{K^2}\right], \\[0.25\baselineskip]
F_{lL}^0=E\left[\frac{(\alpha_l-2)W-\alpha_l L}{K^3}\right].
\end{gather*}

Letting $\alpha_l\to 1$ yields
\[
F_{lW}^0\to
E\left[\frac{-W-L}{K^3}-\tau\frac{W+L}{K^2}\right]
=
-E\left(\frac{1}{K^2}+\frac{\tau}{K}\right),
\qquad
F_{lL}^0\to
E\left[\frac{-W-L}{K^3}\right]
=
-\frac{E}{K^2}.
\]
Along $\alpha_w=0$, we may pass to the limit $(W,L)\to (W_\infty,L_\infty)$, so
\[
F_{lW}^0-F_{lL}^0 \longrightarrow -\frac{E\tau}{K}
\quad\text{as }N\to\infty.
\]

Define $B(W,L):=\frac{L}{K^2}+\tau\frac{L}{K}$. Then the winner's FOC at $\alpha_w=0$ is
\begin{gather*}
    F_w(W,L,0)=B(W,L)\cdot E-\rho
\end{gather*}
Hence, using the product rule and $E_W=-\tau E$, $E_L=0$,
\begin{gather}
F_{wW}^0
=\frac{\partial}{\partial W}\bigl(BE\bigr)
= B_W\,E + B\,E_W
=E\bigl(B_W-\tau B\bigr), \label{eq:FwW0}\\[4pt]
F_{wL}^0
=\frac{\partial}{\partial L}\bigl(BE\bigr)
= B_L\,E + B\,E_L
=E\,B_L. \label{eq:FwL0}
\end{gather}
Writing $B=L K^{-2}+\tau L K^{-1}$ and noting that $K_W,K_L=1$ we have
\begin{gather}
B_W= L\cdot(-2)K^{-3}\cdot K_W + \tau L\cdot(-1)K^{-2}\cdot K_W  -\frac{2L}{K^3}-\tau\frac{L}{K^2}, \label{eq:BW} \\[0.25\baselineskip]
B_L=\Bigl(K^{-2}+L\cdot(-2)K^{-3}\cdot K_L\Bigr)
+\tau\Bigl(K^{-1}+L\cdot(-1)K^{-2}\cdot K_L\Bigr) = \frac{W-L}{K^3}+\tau\frac{W}{K^2}. \label{eq:BL}
\end{gather}
Substituting \eqref{eq:BW}--\eqref{eq:BL} into \eqref{eq:FwW0}--\eqref{eq:FwL0} gives the explicit Jacobian entries:
\begin{gather*}
F_{wW}^0
=E\left(-\frac{2L}{K^3}-\tau\frac{L}{K^2}-\tau\left(\frac{L}{K^2}+\tau\frac{L}{K}\right)\right)
=E\left(-\frac{2L}{K^3}-2\tau\frac{L}{K^2}-\tau^2\frac{L}{K}\right), \label{eq:FwW0-exp} \\[0.25\baselineskip]
F_{wL}^0
=E\left(\frac{W-L}{K^3}+\tau\frac{W}{K^2}\right). \label{eq:FwL0-exp}
\end{gather*}
Since $F_{l\alpha_w}^0<0$ and
\begin{gather*}
    F_{wL}^0-F_{wW}^0 = \left(\frac{1}{K^2}+\tau\frac{W+2L}{K^2}+\tau^2\frac{L}{K}\right)E>0,
\end{gather*}
it follows that $(F_{wL}^0-F_{wW}^0)F_{l\alpha_w}^0<0$.
Since equilibrium is unique when $\alpha_w=0$ for any $N$, $\det J^0$ stays bounded away from $0$ for large $N$:
\begin{equation}\label{eq:det-bounded}
\exists\,c>0,\ \exists\,N_0\ \text{such that}\quad |\det J^0|\ge c \quad\forall\,N\ge N_0.
\end{equation}


Next, $F_{wW}^0$ and $F_{wL}^0$ are given explicitly by \eqref{eq:FwW0-exp}--\eqref{eq:FwL0-exp}.
In particular, for fixed $(W,L)$ these are finite, so there exists $M<\infty$ such that for all large $N$,
\begin{equation}\label{eq:bounded-Fw}
|F_{wL}^0-F_{wW}^0|\le M.
\end{equation}
Combining the expressions derived , \eqref{eq:det-bounded}, and \eqref{eq:bounded-Fw} yields the explicit estimate
\[
\left|
\frac{(F_{wL}^0-F_{wW}^0)\,F_{l\alpha_w}^0}{\det J^0}
\right|
\le
\frac{M}{c}\left|F_{l\alpha_w}^0\right|
=
\frac{M}{c}\cdot \frac{L}{K^2}\cdot \frac{E}{(N-1)^2}.
\]
Because $L/K^2$ and $E$ are bounded along the convergent equilibrium sequence, the right-hand side tends to zero as $N\to\infty$. Therefore,
\[
\frac{(F_{wL}^0-F_{wW}^0)\,F_{l\alpha_w}^0}{\det J^0}\to 0
\quad\text{as }N\to\infty.
\]
and thus,
\begin{gather}\label{eq:reduced}
\lim_{N\to\infty}\bigl(W'(0)+L'(0)\bigr)
=
\lim_{N\to\infty}
\frac{(F_{lW}^0-F_{lL}^0)\,F_{w\alpha_w}^0}{\det J^0}.
\end{gather}

It remains to find the sign of $\det J^0$ and $F_{w\alpha_w}^0$. Let
\[
\det J^0 = E^2\,\widetilde D_N,
\]
where $\widetilde D_N$ is the determinant of the $E$-free matrix. Let $N\to\infty$, equivalently $\alpha_l\to 1$, and define $\widetilde D_\infty:=\lim_{N\to\infty}\widetilde D_N$. A direct algebraic simplification yields
\begin{gather*}\label{eq:Dinf}
\widetilde D_\infty
=
\frac{1}{K^4}
+\frac{\tau(L+2W)}{K^4}
+\frac{\tau^2}{K^2}>0.
\end{gather*}
Consequently, for all sufficiently large $N$,
\[
\widetilde D_N>0
\implies 
\det J^0 = E^2\widetilde D_N>0.
\]
Let $N\to\infty$ so $\alpha_l\to 1$. Then the loser's FOC at $\alpha_w=0$ is
\[
\frac{W+L}{K^2}E=\rho
\implies
E=\rho K.
\]
Combining with the winner's first-order condition gives $B\rho K=K$, i.e.\ $BK=1$. Since $B=\frac{L}{K^2}+\frac{\tau L}{K}$, this implies
\begin{gather*}
    \frac{L}{K}=\frac{1}{1+\tau K}, \\[0.25\baselineskip]
W=K-L=\frac{\tau K^2}{1+\tau K}.
\end{gather*}
By \eqref{eq:fwalpha}, $F_{w\alpha_w}^0=A-\rho$, so dividing by $E>0$ and using $\rho/E=1/K$ (because $E=\rho K$),
\[
\frac{F_{w\alpha_w}^0}{E}
=\frac{A}{E}-\frac{1}{K}.
\]
We also have
\[
A=\frac{1}{W}\cdot\frac{W+L(1-E)}{K}
=\frac{1}{W}\left(1-\frac{L}{K}E\right)
\implies 
\frac{A}{E}=\frac{1}{W}\left(\frac{1}{E}-\frac{L}{K}\right).
\]
Using $E=\rho K$, $\frac{L}{K}=\frac{1}{1+\tau K}$, and $\frac{1}{W}=\frac{1+\tau K}{\tau K^2}$ yields
\[
\frac{F_{w\alpha_w}^0}{E}
=
\frac{1+\tau K}{\tau K^2}\left(\frac{1}{\rho K}-1\right).
\]
Because $W>0$ and $T>0$, we have $E=e^{-\tau W}\in(0,1)$, hence $\rho K=E\in(0,1)$ and thus $\frac{1}{\rho K}-1>0$.
All remaining factors are positive, so as $N\to\infty$
\begin{equation*}\label{eq:signFwalpha}
\frac{F_{w\alpha_w}^0}{E}>0
\implies 
F_{w\alpha_w}^0>0
\end{equation*}
Combining everything meas that for large $N$, the quotient in \eqref{eq:reduced} converges to a strictly negative limit, and hence $\lim_{N\to\infty}\bigl(W'(0)+L'(0)\bigr)<0$.

\paragraph{More Sequencer Revenue.} 

Using our reparameterizations, we can write a loser's equilibrium continuation payoff, which we now denote by $U$, as
\[
U
:=V\Bigl(\frac{L}{K}\Bigr)(1-\alpha_l)E-C\,L(1-\alpha_l)
=V(1-\alpha_l)\Bigl(\frac{L}{K}E-\rho L\Bigr),
\]
and assume the equilibrium path in the $(\alpha_w,N)$-space, given by $(W(\alpha_w;N),L(\alpha_w;N))$ satisfies
\begin{gather*}
    F_w(W,L,\alpha_w)=0, \\[0.25\baselineskip]
    F_l(W,L,\alpha_w)=0.
\end{gather*}
Let $H:=\frac{L}{K}E-\rho L$. Then
\begin{gather*}
    U=V(1-\alpha_l)H \\[0.25\baselineskip]
    U'
    =
    V\bigl[(1-\alpha_l)H'-\alpha_l'H\bigr].
\end{gather*}
Since
\begin{gather*}
    \left(\frac{L}{K}\right)'=\frac{WL'-LW'}{K^2},
\qquad
E'=-T W' E,
\end{gather*}
we obtain
\begin{gather*}
    H'
=
E\left(\frac{WL'-LW'}{K^2}-\tau\frac{L}{K}W'\right)-\rho L'.
\end{gather*}
Substituting, we then have
\[
U'
=
V\left\{
(1-\alpha_l)\left[
E\left(\frac{WL'-LW'}{K^2}-\tau\frac{L}{K}W'\right)-\rho L'
\right]
-\alpha_l'\left(\frac{L}{K}E-\rho L\right)
\right\}.
\]
In equilibrium, $F_l=0$ implies
\[
\frac{W+\alpha_l L}{K^2}E=R
\implies
E=\frac{\rho K^2}{W+\alpha_l L}.
\]
Thus, we have
\begin{gather*}
    \frac{L}{K}E-\rho L
=
\rho L\left(\frac{K-(W+\alpha_l L)}{W+\alpha_l L}\right)
=
\rho L\,\frac{(1-\alpha_l)L}{W+\alpha_l L}.
\end{gather*}
We can write the loser-type's equilibrium utility as
\[
U
=
V(1-\alpha_l)\,\rho L\,\frac{(1-\alpha_l)L}{W+\alpha_l L}
=
C\,\frac{(1-\alpha_l)^2L^2}{W+\alpha_l L}.
\]
Let $\varepsilon:=1-\alpha_l$ and $D:=W+\alpha_l L$. Then
\[
U=C\,\frac{\varepsilon^2L^2}{D}.
\]
Differentiating with respect to $\alpha_w$ gives
\begin{gather*}
    U'
=
C\left(
\frac{2\varepsilon\varepsilon' L^2+2\varepsilon^2 L L'}{D}
-\frac{\varepsilon^2L^2 D'}{D^2}
\right), \\[0.25\baselineskip]
D'=W'+\alpha_l' L+\alpha_l L'.
\end{gather*}
At $\alpha_w=0$,
\begin{gather*}
\varepsilon_N:=1-\alpha_l(0)=\frac{1}{N-1}
\implies
\varepsilon_N'=\varepsilon_N^2, \
\alpha_l'(0)=-\varepsilon_N^2.
\end{gather*}
Substituting into $U'$ yields
\[
U'(0)
=
C\,\varepsilon_N^2
\left(
\frac{2\varepsilon_N L^2+2L \cdot L'(0)}{D_0}
-\frac{L^2\bigl(W'(0)+(1-\varepsilon_N)L'(0)\bigr)}{D_0^2}
\right)
+O(\varepsilon_N^4),
\]
where $D_0=W+(1-\varepsilon_N)L\to S$ as $N\to\infty$. Dividing by $C\varepsilon_N^2$ and letting $N\to\infty$ gives
\[
\lim_{N\to\infty}\frac{U'(0)}{C\varepsilon_N^2}
=
\frac{2L_\infty L_\infty'}{S_\infty}
-\frac{L_\infty^2(W_\infty'+L_\infty')}{S_\infty^2}.
\tag{7}
\]
As $N\to\infty$, $\alpha_l\to1$ and $F_{l\alpha_w}(0)\to0$, so by the Implicit Function Theorem,
\[
\begin{pmatrix}
F_{wW}^\infty & F_{wL}^\infty\\
F_{lW}^\infty & F_{lL}^\infty
\end{pmatrix}
\binom{W_\infty'}{L_\infty'}
=
-
\binom{F_{w\alpha_w}^\infty}{0}.
\]
The second row implies
\[
F_{lW}^\infty W_\infty'+F_{lL}^\infty L_\infty'=0.
\]
At $\alpha_l=1$,
\begin{gather*}
    F_{lW}^\infty=-E\left(\frac{1}{K^2}+\frac{T}{K}\right), \\[0.25\baselineskip]
F_{lL}^\infty=-\frac{E}{K^2},
\end{gather*}
so $L_\infty'=-(1+\tau K_\infty)W_\infty'$. From the limiting $\alpha_w=0$ equilibrium we also have
\[
\frac{L_\infty}{S_\infty}=\frac{1}{1+\tau K_\infty}.
\]
Moreover, since $\lim_{N\to\infty}(W'(0)+L'(0))<0$, it follows that $W_\infty'>0$. Substituting yields
\[
\frac{2L_\infty L_\infty'}{K_\infty}
-\frac{L_\infty^2(W_\infty'+L_\infty')}{K_\infty^2}
=
W_\infty'\left[
-2(1+x)\frac{1}{1+x}
+x\frac{1}{(1+x)^2}
\right],
\]
where we denote $x:=\tau K_\infty>0$. Simplifying,
\[
= W_\infty'\left(-2+\frac{x}{(1+x)^2}\right)
=
-\,W_\infty'\frac{2+3x+2x^2}{(1+x)^2}
<0,
\]
since $W_\infty'>0$ and $x>0$. Therefore, $\lim_{N\to\infty}\frac{U'(0)}{C\varepsilon_N^2}<0$ so it follows that $\lim_{N\to\infty}U'(0)=0^{-}$. In particular, $U'(0)<0$ for all sufficiently large $N$.

\section{Additional Proofs}

\subsection{Multiple Winners: Concavity of Winner's Payoff}

Write the winner's payoff as
\begin{gather*}
    u_w^{(i)}(x;a,L)
    =
    V\cdot\frac{x}{K_w}\left(1-\frac{L}{K}\exp(-\tau K_w)\right)-Cx.
\end{gather*}
Since $(-Cx)''=0$ and $V>0$, it suffices to prove concavity of
\[
s(x):=\frac{x}{K_w}\left(1-\frac{K_l}{K}e^{-\tau K_w}\right)
\]
where $K_w=K_{-w}+x$ and $K=K_l+K_{-w}+x$. A direct differentiation yields the identity
\begin{equation}
s''(x)
=
-\frac{e^{-\tau K_w}}{K_w^3 K^3}\,P(x;K_l,K_{-w},\tau),
\end{equation}
where $P(x;K_{-w},K_l,\tau)=e^{\tau K_w}A(x;K_{-w},K_l)+B(x;K_{-w},K_l,\tau)$
with
\begin{align*}
A(x;K_{-w},K_l)
&\coloneqq 2K_l^3K_{-w}+6K_l^2K_{-w}^2+6K_l^2K_{-w}x+6K_lK_{-w}^3+12K_lK_{-w}^2x+6K_lK_{-w}x^2 \\
&\quad+2K_{-w}^4+6K_{-w}^3x+6K_{-w}^2x^2+2K_{-w}x^3, \\
B(x;K_{-w},K_l,\tau)
&\coloneqq K_l^3K_{-w}^2\tau^2x-2K_l^3K_{-w}^2\tau+2K_l^3K_{-w}\tau^2x^2-2K_l^3K_{-w}\tau x-2K_l^3K_{-w}+K_l^3\tau^2x^3 \\
&\quad+2K_l^2K_{-w}^3\tau^2x-4K_l^2K_{-w}^3\tau+6K_l^2K_{-w}^2\tau^2x^2-6K_l^2K_{-w}^2\tau x-6K_l^2K_{-w}^2 \\
&\quad+6K_l^2K_{-w}\tau^2x^3-6K_l^2K_{-w}x+2K_l^2\tau^2x^4+2K_l^2\tau x^3+K_lK_{-w}^4\tau^2x \\
&\quad-2K_lK_{-w}^4\tau+4K_lK_{-w}^3\tau^2x^2-4K_lK_{-w}^3\tau x-4K_lK_{-w}^3 \\
&\quad+6K_lK_{-w}^2\tau^2x^3-6K_lK_{-w}^2x+4K_lK_{-w}\tau^2x^4+4K_lK_{-w}\tau x^3 \\
&\quad+K_l\tau^2x^5+2K_l\tau x^4+2K_lx^3. 
\end{align*}
Thus $s''(x)\le 0$ if $P(x;K_{-w},K_l,\tau)\ge 0$. Each term in has a nonnegative coefficient, so
\begin{equation*}
A(x;K_{-w},K_l)\ge 0.
\end{equation*}
Moreover, for all $t\ge 0$, $e^t\ge 1+t$ (e.g.\ by convexity of $e^t$ or by the monotonicity of $e^t-1-t$). Applying this with $t=\tau K_w\ge 0$ gives
\begin{equation*}
e^{\tau K_w}\ge 1+\tau K_w.
\end{equation*}
so it follows that
\begin{equation*}\
P(x;K_{-w},K_l,\tau)\ge (1+\tau K_w)A(x;K_{-w},K_l)+B(x;K_{-w},K_l,\tau).
\end{equation*}
Define the right-hand side of \eqref{eq:P-lb} as
\[
R(x;K_{-w},K_l,\tau):=(1+\tau K_w)\cdot A(x;K_{-w},K_l)+B(x;K_{-w},K_l,\tau).
\]
We can write $R(x;K_{-w},K_l,\tau)$ as 
\begin{equation*}
R(x;K_{-w},K_l,\tau)=(K_{-w}+x)^2\,S(x;K_{-w},K_l,\tau),
\end{equation*}
where
\begin{align*}
S(x;K_{-w},K_l,\tau)
&:=K_l^3\tau^2x
+2K_l^2K_{-w}\tau^2x+2K_l^2K_{-w}\tau
+2K_l^2\tau^2x^2+2K_l^2\tau x \\
&\quad+K_lK_{-w}^2\tau^2x+4K_lK_{-w}^2\tau
+2K_lK_{-w}\tau^2x^2+6K_lK_{-w}\tau x+2K_lK_{-w} \\
&\quad+K_l\tau^2x^3+2K_l\tau x^2+2K_lx
+2K_{-w}^3\tau+4K_{-w}^2\tau x+2K_{-w}^2 \\
&\quad+2K_{-w}\tau x^2+2K_{-w}x.
\end{align*}
Every term in $S$ has a nonnegative coefficient, so $S(x;K_{-w},K_l,\tau)$ and thus $R(x;K_{-w},K_l,\tau)\geq0$ as well as $P(x;K_{-w},K_l,\tau)$ are nonnegative. Moreover, if $(K_{-w},K_l)\neq(0,0)$ and $K_w=K_{-w}+x>0$, then the inequality is strict, so $u_w^{(i)}$ is concave in $k_w^{(i)}$.

\subsection{Global Monotonicity of $K( T)$ when $n_w=1$}

Fix $ T\ge 0$. Define $E( T)=\exp\!\bigl(-\lambda T\,k_w\bigr)$ and $K( T)=k_w( T)+n_l\cdot k_l( T)$. The symmetric equilibrium FOCs are then
\begin{gather*}
\rho
=
\left(\frac{n_l \cdot k_l( T)}{K( T)^2}+\lambda T\,\frac{n_l k_l( T)}{K( T)}\right)E( T), \\[0.25\baselineskip]
\rho
=
\frac{k_w( T)+(n_l-1)\cdot k_l( T)}{K( T)^2}\,E( T).
\end{gather*}
Suppressing dependence on $ T$, equating the right-hand sides and cancelling $E( T)>0$ yields
\[
\frac{n_l k_l}{K^2}+\lambda T\,\frac{n_l k_l}{K}
=\frac{k_w+(n_l-1)k_l}{K^2} \implies K=(n_l+1)k_l+\lambda T\,n_l k_l K.
\]
Defining $\theta=\lambda T K$ and $D(\theta)=(1+\theta)n_l+1$, we obtain 
\begin{gather*}
k_l=\frac{K}{D(\theta)}, \\[0.25\baselineskip]
k_w=\frac{(1+n_l \theta)K}{D(\theta)}.
\end{gather*}
Substituting the above expression for $k_w$ into the loser's FOC, we have
\[
K=k_w+(n_l-1)k_l
=\frac{(1+n_l \theta)K}{D(\theta)}+\frac{(n_l-1)K}{D(\theta)}
=\frac{n_l(1+\theta)K}{D(\theta)}.
\]
so we obtain a single equation in $K$ and $ T$:
\begin{gather*}\label{eq:G_eq}
H(K, T)\coloneqq
\frac{n_l(1+\theta)}{K\cdot D(\theta)}\,
\exp\!\left(-\lambda T\,\frac{(1+n_l \theta)K}{D(\theta)}\right)=\rho,
\end{gather*}
Define $G(K, T)\coloneqq H(K, T)-\rho$. Along the equilibrium path $G(K( T), T)=0$,
the implicit-function theorem yields
\begin{gather*}
K'( T)=-\frac{G_ T}{G_K}=-\frac{H_ T}{H_K}.
\end{gather*}
Thus it suffices to show $H_ T<0$ and $H_K<0$ for all $ T>0$.

We first show that $H_ T<0$. Note that
\[
\log H(K, T)
=
\log n_l+\log(1+\theta)-\log K-\log D(\theta)-\lambda T k_w.
\]
Holding $K$ fixed, we have $\partial_ T \theta=\lambda K$ and $\partial_ T D(\theta)=n_l\partial_ T \theta=\lambda n_l K$.
Moreover,
\[
\frac{\partial k_w}{\partial \theta}
=\frac{\partial}{\partial \theta}\left(\frac{(1+n_l \theta)K}{D(\theta)}\right)
=
K\cdot \frac{n_l D(\theta)-(1+n_l \theta)n_l}{D(\theta)^2}
=
K\cdot \frac{n_l^2}{D(\theta)^2},
\]
so we have that
\[
\frac{\partial k_w}{\partial T}
=
\frac{\partial k_w}{\partial \theta}\cdot\frac{\partial \theta}{\partial T}
=
\lambda n_l^2 \frac{K^2}{D(\theta)^2}.
\]
Therefore,
\begin{align*}
\frac{\partial}{\partial T}\log H(K, T)
&=
\frac{\lambda K}{1+\theta}
-\frac{\lambda n_l K}{D(\theta)}
-\lambda k_w
-\lambda T\,\frac{\partial k_w}{\partial T}\\
&=
\frac{\lambda K}{1+\theta}
-\frac{\lambda n_l K}{D(\theta)}
-\lambda\frac{(1+n_l \theta)K}{D(\theta)}
-\lambda T\left(\lambda n_l^2\frac{K^2}{D(\theta)^2}\right).
\end{align*}
Dividing by $\lambda K>0$, we have
\begin{gather*}
\frac{1}{\lambda K}\frac{\partial}{\partial T}\log H(K, T)
=
\frac{1}{1+\theta}
-\frac{n_l}{D(\theta)}
-\frac{1+n_l \theta}{D(\theta)}
-\frac{\theta  n_l^2}{D(\theta)^2}
=
-\frac{\theta}{1+\theta}-\frac{\theta n_l^2}{D(\theta)^2}<0
\end{gather*}
Hence $\partial_ T\log H(K, T)<0$, so $H_ T(K, T)<0$ for all $ T,K>0$.

We then show that $H_K<0$, holding $ T$ fixed. Note that $\partial_K \theta=\lambda T = \theta/K$ and $\partial_K D(\theta)=n_l\partial_K \theta=n_l \theta/K$. Then
\begin{align*}
\frac{\partial}{\partial K}\log H(K, T)
&=
\frac{1}{1+\theta}\frac{\partial \theta}{\partial K}
-\frac{1}{K}
-\frac{1}{D(\theta)}\frac{\partial D}{\partial K}
-\lambda T\,\frac{\partial k_w}{\partial K}\\
&=
\frac{a/K}{1+\theta}
-\frac{1}{K}
-\frac{n_l \theta/K}{D(\theta)}
-\frac{a}{K}\left(\frac{1+n_l \theta}{D(\theta)}+\frac{a n_l^2}{D(\theta)^2}\right).
\end{align*}
Multiplying by $K>0$, we have
\begin{align*}
K\frac{\partial}{\partial K}\log H(K, T)
=
\frac{\theta}{1+\theta}-1
-\frac{n_l \theta}{D(\theta)}
-\frac{\theta(1+n_l \theta)}{D(\theta)}
-\frac{\theta^2 n_l^2}{D(\theta)^2}
=
-\frac{1}{1+\theta}-\theta-\frac{\theta^2 n_l^2}{D(\theta)^2}<0
\end{align*}
so $H_K(K, T)<0$.

\subsection{Multiple Winners: Winner Copies Greater Than Loser Copies}

Write $W=n_wk_w^*$, $L=n_lk_l^*$, and $E=e^{-\tau W}$. Dividing both equations by $V$ and rewriting them in terms of $(W,L)$ yields
\begin{gather*}
\frac CV
=
\frac{n_w-1}{n_w W}
\left(
1-\frac{L}{W+L}E
\right)
+\frac{1}{n_w}
\left(
\frac{L}{(W+L)^2}
+
\tau\,\frac{L}{W+L}
\right)E,
\\[0.25\baselineskip]
\frac CV
=
\frac{W+\frac{n_l-1}{n_l}L}{(W+L)^2}\,E.
\end{gather*}
Subtracting the loser's FOC from the winner's FOC and multiplying by $n_l n_w W (W+L)^2>0$
yields
\begin{equation}
n_l(n_w-1)(W+L)^2 + E\cdot\Xi(W,L)=0,
\end{equation}
where $\Xi(W,L)
=
- n_l n_w (W+L)^2
+ \tau n_l W L (W+L)
+ 2 n_l W L
+ n_l L^2
+ n_w W L$.

\noindent\textbf{Step 3. Use of $E<1$.}
For $n_l>0$, $n_w>1$, and $W+L>0$, the first term in above is strictly positive, so $\Xi(W,L)<0$ and
\begin{gather*}
E=\frac{n_l(n_w-1)(W+L)^2}{-\Xi(W,L)}.
\end{gather*}
Since $\tau>0$ and $W>0$, we have $E=e^{-\tau W}\in(0,1)$, and therefore $E<1$ implies
\[
n_l(n_w-1)(W+L)^2<-\Xi(W,L).
\]
Expanding $\Xi$ and collecting terms yields
\[
0<
W\Bigl(
n_l W
-
n_w L
-
\tau n_l L (W+L)
\Bigr).
\]
As $W>0$, this implies $n_l W>n_w L+\tau n_l L (W+L)>n_w L$. Diving by $n_ln_w$ yields $k_l^*<k_w^*$.

\subsection{Multiple Winners: Winner Utility Greater Than Loser Utility}

Let $K=W+L$ and $E=\exp^{-\tau W}$. Using the equilibrium condition for $\rho$, we have
\begin{align*}
    u_w^*-u_l^* &= (1-\alpha_w)-[(1-\alpha_l)-(1-\alpha_w)]\frac{L}{K}E - [(1-\alpha_w)W+(1-\alpha_l)L]\frac{W+\alpha_lL}{K^2}E \\[0.25\baselineskip]
    &= \frac{1}{n_w}-\left(\frac{1}{n_w}+\frac{1}{n_l}\right)\frac{n_lk_l}{K}E-(k_w-k_l)\cdot\frac{K-k_l}{K^2}E
\end{align*}
Multiplying by $n_w$, it suffices to show that 
\begin{align*}
    1 > \left[\left(1+\frac{n_w}{n_l}\right)\frac{n_lk_l}{n_lk_l+n_wk_w}-n_w(k_w-k_l)\frac{K-k_l}{K^2}\right]E.
\end{align*}
Since $E<1$, it remains to show that the bracketed expression is less than one as well. Note that
\begin{align*}
    \left(1+\frac{n_w}{n_l}\right)\frac{n_lk_l}{n_lk_l+n_wk_w}-n_w(k_w-k_l)\frac{K-k_l}{K^2} = \frac{1}{K^2}[(n_lk_l-n_wk_w+2n_wk_l)K+n_wk_l(k_w-k_l)]
\end{align*}
We claim that $(n_lk_l-n_wk_w+2n_wk_l)K+n_wk_l(k_w-k_l)<K^2$. Dividing by $K$ yields
\begin{gather*}
    n_lk_l+2n_wk_l-n_wk_w+n_wk_l(k_w-k_l) < n_lk_l+n_wk_w
\end{gather*}
Collecting terms and simplifying yields
\begin{gather*}
    \frac{k_l}{K}(k_w-k_l) < 2(k_w-k_l)
\end{gather*}
which holds by construction.

\section{Robustness Checks}


\subsection{Prior Treatment Date as Placebo}

For a properly identified setup, treatment effects should attenuate or disappear when designating a date before the actual treatment date. We move the ``treatment date'' 15 days prior to April 2, 2025.

\begin{table}[H] \centering \footnotesize
\begin{tabular}{@{\extracolsep{2pt}}lcccc}
\\[-2.3ex]\hline
\hline \\[-2.3ex]
\\[-2.3ex] & (1) & (2) & (3) & (4) \\
\\[-2.3ex] & log\,\textsf{RepTXs} & log\,\textsf{RepTXs} & log\,\textsf{Revenue} & log\,\textsf{Revenue} \\
\hline \\[-2.3ex]
\textsf{Post$\times$Treated} & -0.318$^{}$ & -0.347$^{}$ & 0.213$^{}$ & 0.158$^{}$ \\
& (0.314) & (0.294) & (0.144) & (0.131) \\
\textsf{BaseGasL2} & & 0.097$^{}$ & & 0.153$^{**}$ \\
& & (0.071) & & (0.069) \\
\textsf{VolumeDEX} & & 0.167$^{}$ & & 0.069$^{}$ \\
& & (0.121) & & (0.160) \\
\hline \\[-2.3ex]
 Observations & 605 & 605 & 605 & 605 \\
 N. of groups & 5 & 5 & 5 & 5 \\
 $R^2$ & 0.005 & 0.023 & 0.003 & 0.033 \\
\hline
\hline \\[-2.3ex]
\multicolumn{5}{l}{Includes chain and day fixed effects. Standard errors are clustered by chain.} \\
\textit{Note:} & \multicolumn{4}{r}{$^{*}$p$<$0.1; $^{**}$p$<$0.05; $^{***}$p$<$0.01} \\
\end{tabular}
\caption{This table reports two-way fixed-effects estimates using a placebo treatment date of April 2, 2025—15 days prior to the actual Timeboost implementation on April 17, 2025. The sample period spans February 17 to June 17, 2025. The dependent variables are the natural logarithm of repeated transactions (columns 1-2) and the natural logarithm of platform revenue in ETH (columns 3-4). Repeated transactions are defined as bursts of identical transactions (same sender, recipient, value, function selector, and calldata) within 2-second windows. Platform revenue includes gas fees from reverted transactions and, for Arbitrum after the placebo date, hypothetical auction payments. Post×Treated is an indicator equal to one for Arbitrum observations after April 2, 2025. BaseGasL2 is the Layer-2 base gas fee in ETH. VolumeDEX is the natural logarithm of DEX trading volume in USD. All variables except the treatment indicator are standardized (z-scored) prior to estimation. All specifications include chain and day fixed effects. Standard errors are clustered by chain. Under the parallel trends assumption, we should not observe significant treatment effects when using a pre-treatment placebo date.}
\end{table}

\newpage

\subsection{Narrower Study Window}

We narrow the study window by 15 days on both ends, so that the new sample period is from March 2 to June 2, 2025.

\begin{table}[H] \centering \footnotesize
\begin{tabular}{@{\extracolsep{2pt}}lcccc}
\\[-2.3ex]\hline
\hline \\[-2.3ex]
\\[-2.3ex]  & (1) & (2) & (3) & (4) \\
\\[-2.3ex] & log\,\textsf{RepTXs} & log\,\textsf{RepTXs} & log\,\textsf{Revenue} & log\,\textsf{Revenue} \\
\hline \\[-2.3ex]
\textsf{Post$\times$Treated} & -0.564$^{**}$ & -0.587$^{**}$ & 0.964$^{***}$ & 0.924$^{***}$ \\
& (0.242) & (0.238) & (0.044) & (0.052) \\
\textsf{BaseGasL2} & & 0.059$^{}$ & & 0.104$^{}$ \\
& & (0.050) & & (0.067) \\
\textsf{VolumeDEX} & & 0.058$^{}$ & & 0.132$^{}$ \\
& & (0.091) & & (0.149) \\
\hline \\[-2.3ex]
 Observations & 465 & 465 & 465 & 465 \\
 N. of groups & 5 & 5 & 5 & 5 \\
 $R^2$ & 0.016 & 0.020 & 0.060 & 0.081 \\
\hline
\hline \\[-2.3ex]
\multicolumn{5}{l}{Includes chain and day fixed effects. Standard errors are clustered by chain.} \\
\textit{Note:} & \multicolumn{4}{r}{$^{*}$p$<$0.1; $^{**}$p$<$0.05; $^{***}$p$<$0.01} \\
\end{tabular}
\caption{This table reports two-way fixed-effects estimates using a narrower sample window that excludes the first and last 15 days of the original sample period. The sample spans March 2 to June 2, 2025, with Timeboost implemented on April 17, 2025. The dependent variables are the natural logarithm of repeated transactions (columns 1-2) and the natural logarithm of platform revenue in ETH (columns 3-4). Repeated transactions are defined as bursts of identical transactions (same sender, recipient, value, function selector, and calldata) within 2-second windows. Platform revenue includes gas fees from reverted transactions and, for Arbitrum post-Timeboost, auction payments from express lane winners. Post×Treated is an indicator equal to one for Arbitrum observations after April 17, 2025. BaseGasL2 is the Layer-2 base gas fee in ETH. VolumeDEX is the natural logarithm of DEX trading volume in USD. All variables except the treatment indicator are standardized (z-scored) prior to estimation. All specifications include chain and day fixed effects. Standard errors are clustered by chain. This narrower window tests whether results are sensitive to the choice of pre- and post-treatment periods.}
\end{table}

\begin{table}[H] \centering \footnotesize
\begin{tabular}{@{\extracolsep{2pt}}lcccc}
\\[-2.3ex]\hline
\hline \\[-2.3ex]
\\[-2.3ex]  & (1) & (2) & (3) & (4) \\
\\[-2.3ex] & log\,\textsf{RepTXs} & log\,\textsf{RepTXs} & log\,\textsf{Revenue} & log\,\textsf{Revenue} \\
\hline \\[-2.3ex]
\textsf{Intercept} & -11.000$^{*}$ & -11.939$^{*}$ & -3.732$^{}$ & -8.377$^{}$ \\
& (5.681) & (6.153) & (6.958) & (7.223) \\
\textsf{Post} & -0.462$^{*}$ & -0.390$^{}$ & 1.307$^{***}$ & 1.155$^{***}$ \\
& (0.259) & (0.263) & (0.355) & (0.370) \\
\textsf{VolumeDEX} & 1.090$^{***}$ & 1.053$^{***}$ & 0.286$^{}$ & 0.563$^{}$ \\
& (0.283) & (0.315) & (0.345) & (0.366) \\
\textsf{BaseGasMain} & -0.200$^{}$ & -0.264$^{*}$ & 0.355$^{**}$ & 0.376$^{**}$ \\
& (0.149) & (0.145) & (0.169) & (0.175) \\
\textsf{BaseGasArb} & & 169.113$^{***}$ & & -19.931$^{}$ \\
& & (52.415) & & (34.822) \\
\textsf{VolatilityETH} & & -0.114$^{}$ & & -1.926$^{**}$ \\
& & (0.839) & & (0.775) \\
\hline \\[-2.3ex]
 Observations & 123 & 123 & 123 & 123 \\
 $R^2$ & 0.224 & 0.255 & 0.313 & 0.342 \\
 Adjusted $R^2$ & 0.205 & 0.224 & 0.296 & 0.314 \\
\hline
\hline \\[-2.3ex]
\multicolumn{5}{l}{Standard errors are computed using the Newey-West estimator.} \\
\textit{Note:} & \multicolumn{4}{r}{$^{*}$p$<$0.1; $^{**}$p$<$0.05; $^{***}$p$<$0.01} \\
\end{tabular}
\caption{This table reports estimates from a within-Arbitrum analysis using a wider 5-second window to identify repeated transactions. The sample consists of daily observations on Arbitrum from March 2 to June 2, 2025, with Timeboost implemented on April 17, 2025. The dependent variables are the natural logarithm of repeated transactions (columns 1-2) and the natural logarithm of platform revenue in ETH (columns 3-4). Repeated transactions are defined as bursts of identical transactions (same sender, recipient, value, function selector, and calldata) within 5-second windows, compared to the 2-second baseline. Platform revenue includes gas fees from reverted transactions and, post-Timeboost, auction payments from express lane winners. Post is an indicator equal to one for observations after April 17, 2025. VolumeDEX is the natural logarithm of DEX trading volume in USD on Arbitrum. BaseGasMain is the Ethereum mainnet base gas fee in ETH. BaseGasArb is the Arbitrum base gas fee in ETH. VolatilityETH is the daily return volatility of ETH/USD. Variables are not standardized in this specification. Standard errors are computed using the Newey-West estimator to account for heteroskedasticity and autocorrelation. This specification tests whether results are sensitive to the choice of time window for identifying repeated transactions.}
\end{table}

\subsection{Wider Interval for Repeated Transactions}

We widen the interval under which we search for repeated transactions from 2 seconds to 5 seconds.

\begin{table}[H] \centering \footnotesize
\begin{tabular}{@{\extracolsep{2pt}}lcccc}
\\[-2.3ex]\hline
\hline \\[-2.3ex]
\\[-2.3ex]  & (1) & (2) & (3) & (4) \\
\\[-2.3ex] & log\,\textsf{RepTXs} & log\,\textsf{RepTXs} & log\,\textsf{Revenue} & log\,\textsf{Revenue} \\
\hline \\[-2.3ex]
\textsf{Post$\times$Treated} & -0.671$^{***}$ & -0.673$^{***}$ & 0.772$^{***}$ & 0.724$^{***}$ \\
& (0.191) & (0.159) & (0.195) & (0.135) \\
\textsf{BaseGasL2} & & 0.145$^{***}$ & & 0.403$^{***}$ \\
& & (0.044) & & (0.115) \\
\textsf{VolumeDEX} & & 0.207$^{*}$ & & 0.139$^{}$ \\
& & (0.108) & & (0.128) \\
\hline \\[-2.3ex]
 Observations & 605 & 605 & 605 & 605 \\
 N. of groups & 5 & 5 & 5 & 5 \\
 $R^2$ & 0.022 & 0.056 & 0.033 & 0.204 \\
\hline
\hline \\[-2.3ex]
\multicolumn{5}{l}{Includes chain and day fixed effects. Standard errors are clustered by chain.} \\
\textit{Note:} & \multicolumn{4}{r}{$^{*}$p$<$0.1; $^{**}$p$<$0.05; $^{***}$p$<$0.01} \\
\end{tabular}
\caption{This table reports two-way fixed-effects estimates using a wider time interval to detect repeated identical transactions. The sample period spans February 17 to June 17, 2025, with Timeboost implemented on Arbitrum on April 17, 2025. Control Layer-2 blockchains include Optimism, Base, Polygon, and Avalanche. The dependent variables are the natural logarithm of repeated transactions (columns 1-2) and the natural logarithm of platform revenue in ETH (columns 3-4). Repeated transactions are defined as bursts of identical transactions (same sender, recipient, value, function selector, and calldata) within 2-second windows. Platform revenue includes gas fees from reverted transactions and, for Arbitrum post-Timeboost, auction payments from express lane winners. Post×Treated is an indicator equal to one for Arbitrum observations after April 17, 2025. BaseGasL2 is the Layer-2 base gas fee in ETH. VolumeDEX is the natural logarithm of DEX trading volume in USD. All variables except the treatment indicator are standardized (z-scored) prior to estimation. All specifications include chain and day fixed effects. This specification addresses concerns that cluster-robust standard errors may be biased when the number of clusters is small (five chains).}
\end{table}

\begin{table}[H] \centering \footnotesize
\begin{tabular}{@{\extracolsep{2pt}}lcccc}
\\[-2.3ex]\hline
\hline \\[-2.3ex]
\\[-2.3ex] & (1) & (2) & (3) & (4) \\
\\[-2.3ex] & log\,\textsf{RepTXs} & log\,\textsf{RepTXs} & log\,\textsf{Revenue} & log\,\textsf{Revenue} \\
\hline \\[-2.3ex]
\textsf{Intercept} & -9.453$^{}$ & -9.028$^{}$ & 0.566$^{}$ & -6.362$^{}$ \\
& (5.784) & (5.859) & (7.011) & (7.060) \\
\textsf{Post} & -0.542$^{**}$ & -0.416$^{*}$ & 0.783$^{**}$ & 0.527$^{}$ \\
& (0.228) & (0.228) & (0.344) & (0.350) \\
\textsf{VolumeDEX} & 1.015$^{***}$ & 0.894$^{***}$ & 0.084$^{}$ & 0.529$^{}$ \\
& (0.288) & (0.299) & (0.348) & (0.360) \\
\textsf{BaseGasMain} & -0.062$^{}$ & -0.159$^{}$ & 0.608$^{***}$ & 0.663$^{***}$ \\
& (0.152) & (0.149) & (0.182) & (0.191) \\
\textsf{BaseGasArb} & & 190.117$^{***}$ & & -86.322$^{*}$ \\
& & (52.041) & & (50.565) \\
\textsf{VolatilityETH} & & 0.243$^{}$ & & -3.182$^{***}$ \\
& & (0.723) & & (0.830) \\
\hline \\[-2.3ex]
Observations & 121 & 121 & 121 & 121 \\
$R^2$ & 0.262 & 0.315 & 0.229 & 0.313 \\
Adjusted $R^2$ & 0.243 & 0.285 & 0.209 & 0.283 \\
\hline
\hline \\[-2.3ex]
\multicolumn{5}{l}{Standard errors are computed using the Newey-West estimator.} \\
\textit{Note:} & \multicolumn{4}{r}{$^{*}$p$<$0.1; $^{**}$p$<$0.05; $^{***}$p$<$0.01} \\
\end{tabular}
\caption{This table reports two-way fixed-effects estimates of the impact of Timeboost adoption on Arbitrum relative to control Layer-2 blockchains. The sample period spans February 17 to June 17, 2025, with Timeboost implemented on April 17, 2025. The dependent variables are the natural logarithm of repeated transactions (columns 1-2) and the natural logarithm of platform revenue in ETH (columns 3-4). Repeated transactions are defined as bursts of identical transactions (same sender, recipient, value, function selector, and calldata) within 2-second windows. Platform revenue includes gas fees from reverted transactions and, for Arbitrum post-Timeboost, auction payments from express lane winners. Post×Treated is an indicator equal to one for Arbitrum observations after April 17, 2025. BaseGasL2 is the Layer-2 base gas fee in ETH. VolumeDEX is the natural logarithm of DEX trading volume in USD.}
\end{table}

\subsection{Heteroskedasticity and Autocorrelation Robust Standard Errors}

\begin{table}[H] \centering \footnotesize
\begin{tabular}{@{\extracolsep{2pt}}lcccc}
\\[-2.3ex]\hline
\hline \\[-2.3ex]
\\[-2.3ex] & (1) & (2) & (3) & (4) \\
\\[-2.3ex] & log\,\textsf{RepTXs} & log\,\textsf{RepTXs} & log\,\textsf{Revenue} & log\,\textsf{Revenue} \\
\hline \\[-2.3ex]

\textsf{Post$\times$Treated} 
& -0.699$^{***}$ & -0.698$^{***}$ & 0.760$^{***}$ & 0.745$^{***}$ \\
& (0.172) & (0.170) & (0.174) & (0.179) \\

\textsf{BaseGasL2} 
&  & 0.096$^{**}$ &  & 0.142$^{***}$ \\
&  & (0.041) &  & (0.043) \\

\textsf{VolumeDEX} 
&  & 0.159$^{**}$ &  & 0.074 \\
&  & (0.074) &  & (0.061) \\
\hline \\[-2.3ex]
Observations & 605 & 605 & 605 & 605 \\
N. of groups & 5 & 5 & 5 & 5 \\
$R^2$ & 0.025 & 0.041 & 0.037 & 0.063 \\
\hline
\hline \\[-2.3ex]
\multicolumn{5}{l}{Includes chain and day fixed effects. Standard errors are robust (HC3).}  \\
\textit{Note:} & \multicolumn{4}{r}{$^{*}$p$<$0.1; $^{**}$p$<$0.05; $^{***}$p$<$0.01} \\
\end{tabular}
\caption{This table reports two-way fixed-effects estimates using heteroskedasticity-robust (HC3) standard errors instead of cluster-robust standard errors. The treatment group is Arbitrum after April 17, 2025, and the control group consists of other Layer-2 blockchains. The dependent variables are the natural logarithm of repeated transactions (columns 1–2) and the natural logarithm of platform revenue in ETH (columns 3–4). Repeated transactions are defined as bursts of identical transactions within 2-second windows. Platform revenue includes gas fees from reverted transactions and auction payments from express lane winners. BaseGasL2 is the Layer-2 base gas fee in ETH and VolumeDEX is the natural logarithm of DEX trading volume in USD. All specifications include chain and day fixed effects and use HC3 robust standard errors.}
\label{tab:did-result}
\end{table}

\end{document}